\documentclass[journal]{IEEEtran}
\IEEEoverridecommandlockouts

\usepackage{cite}
\usepackage{graphicx}
\graphicspath{{.../FIG/}}
\DeclareGraphicsExtensions{.pdf,.jpeg,.png,.webp,.jpg}
\usepackage{stfloats}
\usepackage[hidelinks]{hyperref}
\usepackage{amsmath,mathtools}
\usepackage{amsfonts,amssymb}
\usepackage{bm}
\usepackage{algorithm}
\usepackage{algorithmic}

\usepackage{array}
\renewcommand\arraystretch{1}
\usepackage{cuted} 
\usepackage{booktabs}
\usepackage{multirow}
\usepackage{tabularx}
\usepackage{longtable}
\usepackage{tikz}
\usetikzlibrary{arrows.meta, positioning, shapes.geometric}
\usepackage{mathrsfs}

\usepackage{url}

\usepackage{amsthm}

\newtheorem{theorem}{Theorem}[section]
\newtheorem{lemma}[theorem]{Lemma}

\newtheorem{corollary}[theorem]{Corollary}

\newtheorem{remark}[theorem]{Remark}

\usepackage{color}
\definecolor{BLUE}{RGB}{0,114,189}
\definecolor{RED}{RGB}{217,83,25}
\definecolor{YELLOW}{RGB}{237,177,32}
\definecolor{PURPLE}{RGB}{126,47,142}

\usepackage{lipsum}
\usepackage{amssymb}
\usepackage{booktabs}
\usepackage{enumerate}
\usepackage{setspace}
\usepackage{float}

\usepackage[caption=false,font=normalsize,labelfont=sf,textfont=sf]{subfig}
\usepackage{textcomp}
\usepackage{verbatim}

\usepackage{subfiles}

\usepackage{soul}
\usepackage{textcomp}
\usepackage{xcolor}
\definecolor{myblue}{RGB}{0,112,192}

\begin{document}
\bstctlcite{bstcontrol}

\title{Impacts of Heterogeneous Grid-Forming Devices on Power System Dynamics Quantified by DW Shells \vspace{3mm}} 

\author{~Liangxiao~Luo, Linbin~Huang, Hangyu~Chen, Ruohan~Leng, Zhixian~Hou, Kehao~Zhuang, and Huanhai~Xin 
\vspace{-7mm}
\thanks{This work was supported by the Smart Grid National Science and Technology Major Project of China under Grant 2026ZD0812903. 

The authors are with the College of Electrical Engineering, Zhejiang University, Hangzhou 310027, China. (E-mail: hlinbin@zju.edu.cn).}
}
\maketitle

    \begin{abstract}

    The concept of grid-forming (GFM) converters has gained great attention in the past years. However, it remains challenging to analyze and quantify the impacts of heterogeneous GFM devices (e.g., GFM energy storage systems, GFM wind turbines, GFM HVDC stations) on power system dynamics, especially when taking into account the complex interaction between GFM converters and grid-following (GFL) converters. To this end, this paper focuses on the decentralized and scalable stability analysis of power systems containing both GFM and GFL converters, where we use Davis–Wielandt (DW) shells to characterize the dynamics of the converters and the power grid. In particular, we analytically derive how integrating heterogeneous GFM converters affects the DW shell of the power grid and therefore the system stability. Our approach does not require the detailed parameters or control schemes of the GFM converters; instead, we define the local passivity and imaginary-axis indices of GFM converters to compactly describe their characteristics. These two indices can be conveniently obtained by testing a GFM converter and greatly simplify the stability analysis and computation when handling large-scale power systems.

    \end{abstract}

\vspace{-1mm}

\section{Introduction}
    \IEEEPARstart{M}ODERN power systems feature the large-scale integration of heterogeneous converter-based resources, including renewable power plants, energy-storage systems, HVDC links, and power-electronic loads, among others \cite{Cheng2023real,Hatziargyriou2021Def,Milano2018lowinertia,li2022revisiting}. 
    As their penetration increases, the small-signal stability of modern power systems is increasingly determined by the dynamic interactions between converters and the network.

    Grid-following (GFL) and grid-forming (GFM) converters represent two major paradigms of converters. GFL converters are currently widely deployed in practice. They rely on phase-locked loops to synchronize with the external grid voltage~\cite{2020huangGridsyn} and are therefore sensitive to the power network characteristics seen at their terminals. In weak grids, converter--network interactions may give rise to poorly damped sub- and super-synchronous oscillations.
    By contrast, GFM converters can establish local voltage and frequency references and are expected to support the operation of converter-dominated systems~\cite{2017Achieving}.
    The nature and extent of this support, however, are not uniform across different types of GFM converters, as their dynamic behaviors vary not only across applications, such as grid-forming wind generators, energy-storage systems, and HVDC stations, but also across control schemes, including virtual synchronous generator (VSG) control, droop control, virtual oscillator control, and matching control~\cite{Rosso2021GFM,Wu2016VSG,Johnson2016VOC,Matevosyan2019GFM, Huang2017ViSynC}. Quantifying the stability support provided by these heterogeneous GFM devices is therefore challenging, especially when they are impacting the system simultaneously. 

    Existing studies have assessed the stability support of GFM converters using detailed dynamic models. For instance, state-space and impedance models can be used to retain the control dynamics of individual GFM devices and relate their parameters to closed-loop modes or converter-network interactions \cite{Gu2021Imp,Niu202state}. Such analyses are suitable when the control structure and parameters of each device are available. However, they become difficult to apply to large systems containing heterogeneous GFM devices, especially with proprietary control schemes.
    GFM support has also been assessed through its enhancement of grid strength. In~\cite{yang2021placing,xin2025howmany}, the generalized short-circuit ratio (gSCR) is used to derive stability conditions by requiring the gSCR to be large enough, which also investigated how GFM converters increase gSCR and provided useful guidance for the GFM placement and capacity allocation. However, these methods usually describe the GFM contribution through a static or quasi-static voltage-source equivalence and therefore do not distinguish devices with different frequency-domain terminal dynamics. Impedance-based formulations can retain these dynamics by incorporating the full multi-input multi-output (MIMO) terminal impedance of converters \cite{Henderson2022SCR,Lamrani2027GFMPlacement}. Although the required terminal impedance can be obtained from black-box measurements, impedance approaches require constructing and analyzing a high-order system impedance matrix. This becomes intractable and computationally demanding as the number of converters increases.
    It is thus favorable to have a low-dimensional and compact local description that quantifies how each heterogeneous GFM device contributes to the system stability, but such a local description is still missing.

    In recent years, the concept of decentralized stability certificates has gained considerable attention, which offers a scalable alternative for analyzing large-scale converter-dominated power systems. Certificates such as passivity, small-gain, and small-phase conditions all analyze stability using local device properties and network-side information, without constructing the full closed-loop system model~\cite{Harnefors2016passivity,huang2024gain,chen2024phase,Wang2024phase}. Graphical approaches based on the scaled relative graph (SRG) and the Davis-Wielandt (DW) shell further provide unified gain-phase descriptions of the device-grid interactions~\cite{Baron2026SRGs,feng2025unified,huang2025DW,leng2026operating}. However, these methods generally put GFM converters on the device side when partitioning the closed-loop system, and their contribution to the power grid strength cannot be explicitly reflected or efficiently computed when multiple heterogeneous GFM converters are considered. It is even more challenging to quantify how heterogeneous GFM converters impact the stability of GFL converters using existing methods.

    To fill these gaps, this paper develops a geometric method to quantify the stability support of heterogeneous GFM converters. Our approach characterizes each GFM device through local passivity and imaginary-axis indices, which can be conveniently obtained from its black-box admittance model. The GFM dynamics are then fused into the power network model, which interacts with the remaining GFL converters. We rigorously prove that such a fusion process requires only knowing the two local passivity and imaginary-axis indices. To be specific, we derive bounds for the DW shell of the equivalent power network which includes the GFM dynamics via Schur complement, and we find that such bounds depend only on the proposed passivity and imaginary-axis indices of GFM converters. In this manner, we do not need to know the detailed models to capture the stability impact of heterogeneous GFM converters. Moreover, the resulting DW shell envelope yields decentralized stability certificates to analyze how GFL converters interact with the power network and the GFM converters. In summary, our approach avoids dealing with the detailed frequency-domain admittance model of GFM converters when analyzing the system-level dynamics, and is capable of handling large-scale converter-dominated power systems where heterogeneous GFM converters are installed to improve the system stability.

    The remainder of this paper is organized as follows. Section~\ref{sec:system-model} presents the system model and the equivalent power network reformulation. Section~\ref{sec:geometric-conditions} introduces the DW shell and geometric stability conditions. Section~\ref{sec:gfm-support-analysis} derives the envelope of the DW shell of the equivalent power network, the bounds obtained from local GFM indices, and the resulting stability certificates. Section~\ref{sec:case-studies} presents the case studies and time-domain simulations. Section~\ref{sec:conclusions} concludes the paper.

   \section{System Modeling and Equivalent Network Reformulation}\label{sec:system-model}

   This section presents the converter and network admittance models. To explicitly analyze the stabilizing contribution of the GFM converters, we partition the model according to the GFL and GFM nodes and apply the Schur complement to obtain the equivalent network which includes the GFM dynamics.

\vspace{-2mm}
   \subsection{Modeling of Converter and Network Dynamics}
    Consider a power system that consists of $n$ GFL converters, $m$ GFM devices, $k$ interior network nodes, and an infinite bus, as illustrated in Fig.~\ref{fig:hybrid-system-topology}.
    Although our approach can handle any types of line dynamics, for simplicity of illustration, line resistances are ignored in what follows, and the network is modeled by inductive dynamics.
    The dynamics of a grid line connecting Nodes $i$ and $j$ are given by~\cite{dong2019small}:
   \begin{equation}\begin{bmatrix}
   \Delta I_{x,ij} \\
   \Delta I_{y,ij}\end{bmatrix}=
   B_{ij}{R}(s)
   \begin{bmatrix}
   \Delta U_{x,i}-\Delta U_{x,j} \\
   \Delta U_{y,i}-\Delta U_{y,j}
   \end{bmatrix},\label{eq:line-dynamics}
   \end{equation}
   where \(R(s)= \left[\begin{smallmatrix}s & -\omega_0\\\omega_0 & s
   \end{smallmatrix}\right]^{-1}\) is the inductance rotation matrix which represents the admittance dynamics of an inductive line with unit susceptance, and $B_{ij}=1/(L_{ij}\omega_0)$ is the line susceptance; in a global $xy$ reference frame,
   $\begin{bmatrix}\Delta I_{x,ij} , \Delta I_{y,ij}\end{bmatrix}^\top$ and $\begin{bmatrix}\Delta U_{x,i} , \Delta U_{y,i}\end{bmatrix}^\top$ denote the current vector from Node $i$ to Node $j$ and the voltage vector at Node $i$, respectively. 
   The infinite bus can be considered as the ground node in small-signal analysis. Hence, let
   \begin{equation}
    \bm{Q}= \begin{bmatrix}\bm{Q}_1 & \bm{Q}_2 \\\bm{Q}_3 & \bm{Q}_4\end{bmatrix} \in \mathbb{R}^{(n+m+k)\times(n+m+k)}
    \end{equation} 
    be the grounded Laplacian matrix of the electrical network which can be calculated by $\bm{Q}_{ij} = -B_{ij}(i \ne j)$ and $\bm{Q}_{ii} = \sum_{\substack{j=1, j \ne i}}^{n+m+k} B_{ij} + B_{i,n+m+k+1}.$ By performing Kron reduction, we eliminate the interior nodes and obtain the Kron-reduced Laplacian matrix as
    \begin{equation} \label{eq:kron-reduction}
    \bm{Q}_{\mathrm{red}} = \bm{Q}_1 - \bm{Q}_2 \bm{Q}_4^{-1} \bm{Q}_3,
     \end{equation} 
    where $\bm{Q}_1 \in \mathbb{R}^{(m+n)\times(m+n)},\bm{Q}_2 \in \mathbb{R}^{(m+n)\times k}, \bm{Q}_3 \in \mathbb{R}^{k\times(m+n)},\bm{Q}_4 \in \mathbb{R}^{k\times k}.$ Combining~\eqref{eq:line-dynamics} and~\eqref{eq:kron-reduction}, the  Kron-reduced network dynamics can be expressed as:
   \begin{equation}
   \Delta \bm{I} = \bm{Y}_{\text{Grid}}(s) \Delta \bm{U}:=[(\bm{Q}_{\text{red}} \otimes R(s)] \Delta \bm{U},
   \label{eq:kron-reduced-network-admittance}
   \end{equation}
   where \(\Delta\bm I:=[\Delta I_{x,1},\Delta I_{y,1},\ldots, \Delta I_{x,n+m},\Delta I_{y,n+m}]^{\top}\) is the stacked current injection vector of the converters, and \(\Delta\bm U:=[\Delta U_{x,1},\Delta U_{y,1},\ldots, \Delta U_{x,n+m},\Delta U_{y,n+m}]^{\top}\) is the corresponding terminal voltage vector. The subscripts \(x\) and \(y\) refer to the global \(xy\) reference frame, and \(\otimes\) denotes the Kronecker product.
    \begin{figure}[!t]
	\centering
	\includegraphics[width=3.3in]{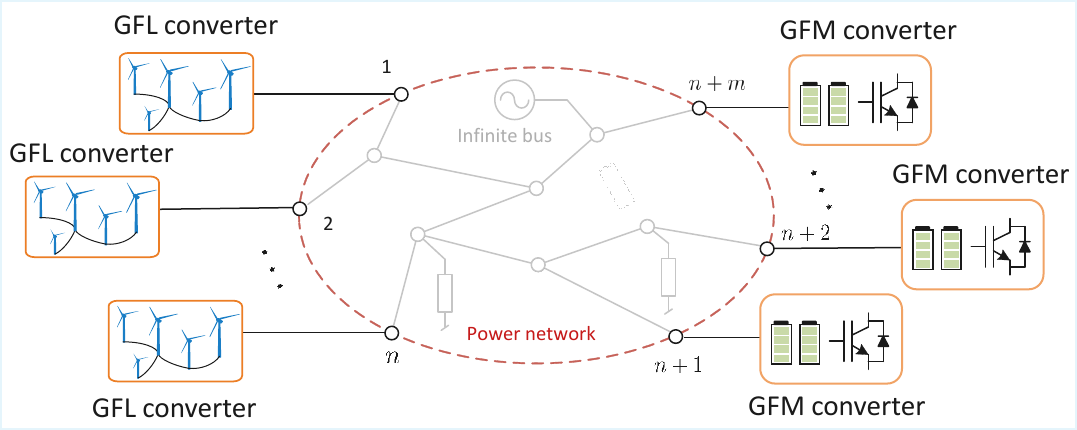}
	\vspace{-1mm}
	\caption{Illustration of a converter-dominated power system.}
	\vspace{-1mm}
	\label{fig:hybrid-system-topology}
    \end{figure}

    \begin{figure}[!t]
    	\centering
        \vspace{0.5mm}
    	\includegraphics[width=3.0in]{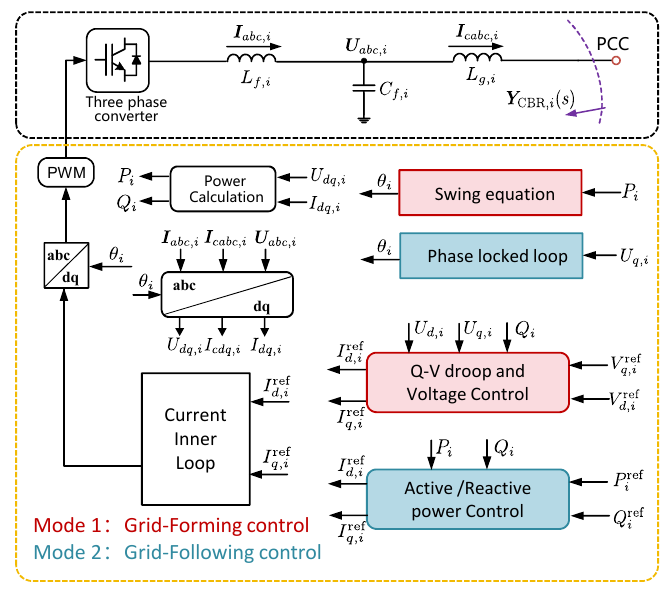}
    	\vspace{-2mm}
    	\caption{A grid-connected three-phase power converter.  Mode 1:  GFM control. Mode 2: GFL control.}
    	\vspace{-3mm}
	\label{fig:converter-control-structure}
      \end{figure}
      
   We next investigate the dynamics of the converters. Fig. \ref{fig:converter-control-structure} shows a three-phase converter connected to the ac grid through an LCL filter. The converter can be operated in GFL mode or GFM mode. 
   The linearized model of the $i$-th converter-based resource (CBR) is represented by a $2\times2$ admittance (transfer function) matrix $\bm{Y}_{\text{CBR},i}(s)$ in the global $xy$-frame:
   \begin{equation}
   -
   \begin{bmatrix}
   \Delta I_{x,i} \\
   \Delta I_{y,i}
    \end{bmatrix}=
    S_{i}\bm{Y}_{\text{CBR},i}(s)
    \begin{bmatrix}
    \Delta U_{x,i} \\
    \Delta U_{y,i}
    \end{bmatrix},
    \label{eq:converter-admittance}
    \end{equation} 
      where $\begin{bmatrix}\Delta I_{x,i} , \Delta I_{y,i}\end{bmatrix}^\mathrm{T}$ and $\begin{bmatrix}\Delta U_{x,i} , \Delta U_{y,i} \end{bmatrix}^\mathrm{T}$ denote the perturbations in the current output and terminal voltage of Converter \(i\) in the global and $S_{i}$ is the capacity ratio of the $i$-th node's rated capacity to the base capacity of per-unit calculation. The detailed derivation of such admittance matrix has been well studied in the literature, e.g., \cite{yang2021placing} and~\cite{huang2024gain}. Then, we extend \eqref{eq:converter-admittance} to include the dynamics of all converters:
    \begin{equation} \label{eq:aggregated-converter-admittance}
    \Delta \bm{I}_{xy}=-(\bm{S}_{\rm{B}}\otimes I_{2}) \bm{Y}_{\rm{CBR}}(s) \Delta \bm{U}_{xy},
    \end{equation}
    where $\bm{S}_{\rm{B}} = \text{diag}\{S_{1},\cdots S_{n+m}\}$ is the capacity ratio matrix; $\bm{Y}_{\rm{CBR}}(s)=\text{diag}\{\bm{Y}_{\text{CBR},1}(s),\cdots,\bm{Y}_{\text{CBR},n+m}(s)\}$ is block-diagonal and represents the dynamics of all converters; \(\operatorname{diag}\{\cdot\}\) denotes a block-diagonal matrix constructed from its arguments, and \(I_k\) denotes the \(k\times k\) identity matrix.
    Note that here we ignore the static angle differences of the converters, because they will not affect the system-level decentralized stability analysis, as proved in~\cite{huang2024gain}.
   
    Combining the converter dynamics in \eqref{eq:aggregated-converter-admittance} and the power network dynamics in \eqref{eq:kron-reduced-network-admittance}, the closed-loop interconnection of the converters and the power network can be written as
    \begin{equation} \label{eq:hybrid-system-characteristic} 
    (\bm{S}_{\rm B}\otimes I_{2})\bm{Y}_{\mathrm{CBR}}(s) \#\bm{Y}_{\mathrm{Grid}}^{-1}(s),
    \end{equation}
    where \(\#\) denotes the feedback interconnection.
     
    To simplify the network representation, we rescale the converter and network dynamics following~\cite{huang2024gain}. The resulting equivalent closed-loop interconnection is
    \begin{equation} \label{eq:rescaled-interconnection}
        \widetilde{\bm{Y}}_{\text{CBR}}(s) \#  \widetilde{\bm{Y}}_{\rm Grid}^{-1}(s),
    \end{equation}
    where $\widetilde{\bm Y}_{\rm Grid}(j\omega) = {\bm S}^{-\frac{1}{2}}_{\rm{B}} {\bm Q}_{\rm red} {\bm S}^{-\frac{1}{2}}_{\rm{B}} \otimes I_2$ and the $i$-th block of $\widetilde{\bm Y}_{{\rm CBR}}(s)$ is denoted by $\widetilde{\bm Y}_{{\rm CBR},i}(s)$, which captures the dynamics of the \text{$i$-th} converter as
    \begin{equation}\label{eq:rescaled-converter-admittance}
        \widetilde{\bm Y}_{{\rm CBR},i}(s) := {\bm Y}_{{\rm CBR},i}(s){R}^{-1}(s) ,\; i \in \{1,..., n+m\}.
    \end{equation} 
    This rescaling absorbs the nductance rotation matrix $R(s)$ into the converter admittance, so that the rescaled network matrix becomes a constant positive-definite matrix.

    \vspace{-3mm}

     \subsection{Equivalent Network Reformulation}
     
     Under the conventional partition in \eqref{eq:rescaled-interconnection}, all converter admittances remain on the converter side, whereas \(\widetilde{\bm Y}_{\rm Grid}(s)\) represents the network. The stability assessment of each GFL converter therefore uses only \(\widetilde{\bm Y}_{\rm Grid}(s)\), and the contribution of the GFM converters is not incorporated into this assessment.
     
     To quantify the stability support provided by the GFM devices, we fuse their dynamics into the network. As illustrated in Fig.~\ref{fig:dynamic-repartitioning}, this re-partitioning yields a reduced interconnection between the GFL converters and an equivalent network. The resulting characteristic equation is obtained below. Firstly, we rewrite \eqref{eq:rescaled-interconnection} according to the GFL and GFM nodes:
    \begin{equation} \label{eq:partitioned-hybrid-characteristic}
    \det \left(\begin{bmatrix}
    \widetilde{\bm{Y}}_{\rm{GFL}}(s) &  \\
    & \hspace{-3mm} \widetilde{\bm{Y}}_{\rm{GFM}}(s)
    \end{bmatrix} +\begin{bmatrix}
    \bm{Q}_{A} &   \bm{Q}_{B} \\
    \bm{Q}_{C} &   \bm{Q}_{D}
    \end{bmatrix} \otimes I_{2}\right) =0,
    \end{equation}
    where $\det(\cdot)$ denotes the determinant; $\bm Q_A$, $\bm Q_B$, $\bm Q_C$, and $\bm Q_D$ are obtained by partitioning ${\bm S}_{\rm B}^{-1/2}{\bm Q}_{\rm red}{\bm S}_{\rm B}^{-1/2}$ according to the GFL and GFM nodes. Their dimensions are $n\times n$, $n\times m$, $m\times n$, and $m\times m$, respectively. Likewise, $\widetilde{\bm Y}_{\mathrm{GFL}}(s)$ and $\widetilde{\bm Y}_{\mathrm{GFM}}(s)$ are the GFL and GFM blocks of $\widetilde{\bm Y}_{\mathrm{CBR}}(s)$.
    
    Then, by applying the Schur complement, Eq.~\eqref{eq:partitioned-hybrid-characteristic} can be equivalently expressed as
    \begin{equation} \label{eq:schur-complement-factorization}
    \begin{aligned}
     &\det(\bm{C}(s))\, \det \Big( \widetilde{\bm{Y}}_{\mathrm{GridC}}
     + \widetilde{\bm{Y}}_{\rm{GFL}}(s) \Big)=0,
    \\
    &\,\bm{C}(s) =\widetilde{\bm{Y}}_{\rm{GFM}}(s) + \bm{Q}_{D} \otimes I_{2}    ,
    \end{aligned}
    \end{equation}
    where $\widetilde{\bm{Y}}_{\mathrm{GridC}}(s)$ characterizes the equivalent network dynamics, including the network dynamics and GFM dynamics:
     \begin{equation}
\label{eq:equivalent network-admittance}
\resizebox{0.98\linewidth}{!}{$
\begin{aligned}
\widetilde{\bm{Y}}_{\mathrm{GridC}}(s)
=&\;
\bm{Q}_A \otimes I_2
-
(\bm{Q}_B\bm{Q}_{D}^{-1/2}\otimes I_2)
\bm{S}_{\mathrm{GFM}}(s)
(\bm{Q}_D^{-1/2}\bm{Q}_C\otimes I_2),
\\
\bm{S}_{\mathrm{GFM}}(s)
=&\;
\left[
I_{2m}
+
(\bm{Q}_{D}^{-1/2}\otimes I_{2})
\widetilde{\bm{Y}}_{\mathrm{GFM}}(s)
(\bm{Q}_{D}^{-1/2}\otimes I_{2})
\right]^{-1}. 
\end{aligned}
$}
\end{equation}
     
      The above reformulations~\eqref{eq:schur-complement-factorization} and~\eqref{eq:equivalent network-admittance} fuse the GFM dynamics into the network representation. The resulting matrix $\widetilde{\bm Y}_{\mathrm{GridC}}(s)$ is the network admittance seen by the GFL converters and includes both the original network dynamics and the embedded GFM dynamics. Eq.~\eqref{eq:schur-complement-factorization} decomposes the characteristic equation into two parts. The part $\det(\bm C(s))$ describes the GFM subsystem, in which the $m$ GFM converters are interconnected through the network submatrix $\bm Q_D$. The remaining part $\det ( \widetilde{\bm{Y}}_{\mathrm{GridC}}
     + \widetilde{\bm{Y}}_{\rm{GFL}}(s) )$ characterizes the interaction between the GFL converters and the equivalent network (embedding GFM converters), which is the focus of this paper. 

     \begin{remark}
     The subsystem associated with $\bm C(s)$ is assumed to be stable, that is, $\det(\bm C(s))=0$ has no poles in the closed right-half plane. This assumption requires that GFM converters should be stable when operating in parallel, and it is often considered during the design phase of GFM converters. This paper will not focus on this problem since we aim at quantifying the impact of GFM converters at the system level.
     \end{remark}

     Under the above assumption, the poles of the closed-loop system that are associated with the GFL-GFM and GFL-network interactions should satisfy
     \begin{equation}
      \label{eq:reduced-characteristic}
       \det\left(
         \widetilde{\bm{Y}}_{\mathrm{GridC}}(s)
          +
         \widetilde{\bm{Y}}_{\mathrm{GFL}}(s)
       \right)=0,
      \end{equation}
     which describes the interconnection $\widetilde{\bm Y}_{\rm GFL}(s)\# \widetilde{\bm Y}_{\rm GridC}^{-1}(s)$ and the stability is governed by the coupled dynamics of the GFL converters, the network, and the GFM converters. Direct analysis requires constructing the high-dimensional equivalent network admittance matrix and evaluating its interaction with all GFL converters. However, its complexity grows with the system size and number of GFM devices, motivating a modular and scalable assessment based on condensed local information.

    \begin{figure}[!t]
    	\centering
        \vspace{0mm}
    	\includegraphics[width=3.0in]{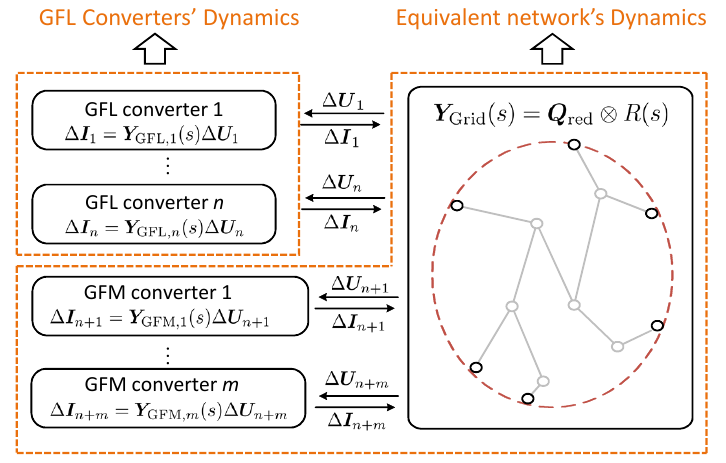}
    	\vspace{-1mm}
    	\caption{Illustration of the re-partitioning of network dynamics and converters' dynamics, where the GFM converter dynamics are fused into the network.}
    	\vspace{-3mm}
	\label{fig:dynamic-repartitioning}
      \end{figure}
      
    \section{Geometric Decentralized Stability Conditions Based on DW Shells}\label{sec:geometric-conditions}

    In this section, we introduce the DW shell and its projections, including numerical range and $x$-$z$ graph. Building on these geometric concepts, we then present decentralized stability conditions for feedback interconnections.

   \subsection{Geometric Characterization of Complex Matrices}

     The Davis-Wielandt (DW) shell provides a geometric description of a complex matrix~\cite{wielandt1955eigenvalues,davis1968shell,zhang2025phantom,lestas_DW,Li_DW}, and we briefly introduce below how it can be used in stability analysis.
     For $A \in \mathbb{C}^{n\times n}$, its DW shell is defined by
     \begin{equation}\label{eq:DW shell-definition} DW(A)=\{(x^*Ax,\|Ax\|^2): x\in \mathbb{C}^n, \|x\|=1\},
     \end{equation}
     where $x^*$ is the conjugate transpose of $x$ and $\|\cdot\|$ is the two-norm. As illustrated in Fig.~\ref{fig:DW shell-projections}, the projection of $DW(A)$ onto the $x$--$y$ plane is the \textbf{numerical range}
    \begin{equation}\label{eq:numerical-range-definition}
      W(A)=\{x^*Ax: x\in \mathbb{C}^n, \|x\|=1\}.
     \end{equation}

     For a complex scalar \(a\), its gain and phase are uniquely given by
     \(|a|\) and \(\angle a\), respectively. For a complex matrix \(A\in\mathbb{C}^{n\times n}\), the \textbf{gain} can be characterized by the interval between its minimum and maximum singular values
    \begin{equation} 
    \sigma(A) = \left[ \sigma_{\min}(A),\, \sigma_{\max}(A) \right],
    \end{equation}
     where \(\sigma_{\min}(A)\) and \(\sigma_{\max}(A)\) denote the minimum and maximum singular values of \(A\), respectively. If \(0\notin W(A)\), the matrix \(A\) is said to be sectorial, and its \textbf{phase} can be characterized by the arguments of the numerical range~\cite{chen2024phase}
     \begin{equation}
     \phi(A)= \left\{ \angle z\mid z\in W(A) \right\}. \end{equation}
     For a sectorial matrix, this phase set is the interval \(\phi(A)=[\phi_{\min}(A),\phi_{\max}(A)]\), where \(\phi_{\min}(A)\) and \(\phi_{\max}(A)\) are determined by the two supporting rays of \(W(A)\) with respect to the positive real axis~\cite{Wang2024phase}.
    
     In~\cite{huang2025DW}, the projection of the DW shell onto the $x$–$z$ plane is defined as the $\bf x$–$\bf z$ \textbf{graph} $P(A)$, which enables complementary stability conditions with the numerical range, denoted by
     \begin{equation}\label{eq:xz-graph-definition}
     P(A)=\{(\Re(x^*Ax),\|Ax\|^2): x\in \mathbb{C}^n, \|x\|=1\},
     \end{equation}
     where \(\Re(\cdot)\) is the real part of a complex number. We next introduce the basic stability conditions enabled by the above concepts, which are essential for the subsequent analysis of heterogeneous GFM converters.

     \vspace{-1mm}
    
    \subsection{Geometric Stability Conditions}

    \begin{figure}[!t]
	\centering
	\includegraphics[width=3.1in]{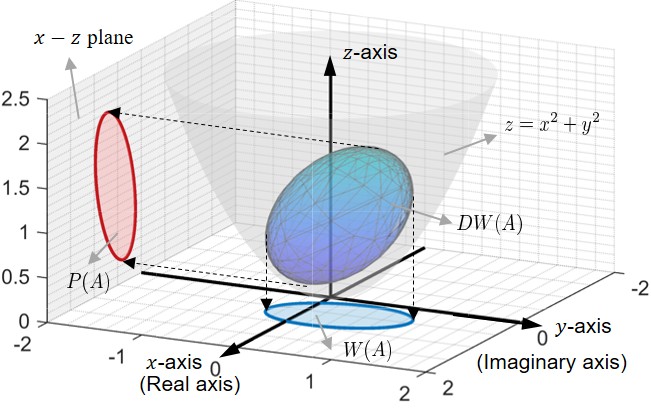}
	\vspace{-2mm}
	\caption{The DW shell $DW(A)$ and its projections onto the $x$-$y$ and $x$-$z$ planes, which are numerical range and $x$-$z$ graph, respectively.}
	\vspace{2mm}
	\label{fig:DW shell-projections}
    \end{figure}

     \setlength{\unitlength}{1mm}
    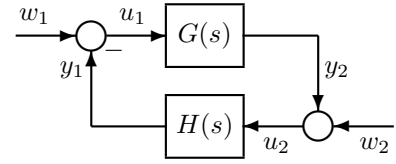
\begin{figure}[!t]
    \begin{center}
    \begin{picture}(50,20)
    \thicklines 
    \put(0,17){\vector(1,0){8}} \put(10,17){\circle{4}}
    \put(12,17){\vector(1,0){8}} \put(20,13){\framebox(10,8){$G(s)$}}
    \put(30,17){\line(1,0){10}} \put(40,17){\vector(0,-1){10}}
    \put(38,5){\vector(-1,0){8}} \put(40,5){\circle{4}}
    \put(50,5){\vector(-1,0){8}} \put(20,1){\framebox(10,8){$H(s)$}}
    \put(20,5){\line(-1,0){10}} \put(10,5){\vector(0,1){10}}
    \put(5,10){\makebox(5,5){$y_1$}} \put(40,10){\makebox(5,5){$y_2$}}
    \put(0,17){\makebox(5,5){$w_1$}} \put(45,0){\makebox(5,5){$w_2$}}
    \put(13,17){\makebox(5,5){$u_1$}} \put(32,0){\makebox(5,5){$u_2$}}
    \put(10,10){\makebox(6,10){$-$}}
    \end{picture}
    \vspace{-2mm}
    \caption{A standard closed-loop system $G(s)\#H(s)$.}
    \vspace{-5mm}
    \label{fig:feedback-interconnection}
    \end{center}
    \end{figure}

    Consider the feedback interconnection \(G(s)\#H(s)\) shown in Fig.~\ref{fig:feedback-interconnection}, where \(G,H\in\mathcal{RH}_{\infty}^{m\times m}\) are real, rational, proper, and stable transfer function matrices. The following DW shell separation condition provides a geometric interpretation of closed-loop stability.
    
    \begin{lemma}[Geometric Stability Condition based on DW Shell Separation~\cite{huang2025DW,feng2025unified}]
     \label{thm:geometric-centralized-stability} 
     The closed-loop system \(G(s)\#H(s)\) is stable if, for each \(\omega\in[0,\infty)\),
     \begin{equation}
     DW\bigl(G(j\omega)\bigr)
      \cap
      DW\!\left(
      -\tfrac{1}{\tau}H^{-1}(j\omega)
       \right)
          =
       \varnothing,
        \forall \tau\in(0,1],
      \label{eq:centralized-DW shell-separation}
      \end{equation}
      that is, the DW shell of \(G(j\omega)\) is separated from the DW shell of
    \(-H^{-1}(j\omega)\) scaled by \(1/\tau\).
      \end{lemma}
    When $G(s)$ is block-diagonal (aligned with the setting of $\widetilde{\bm{Y}}_{\text{CBR}}(s)$ in~\eqref{eq:rescaled-interconnection}), i.e., \(G(s)=\operatorname{diag}\{G_1(s),\ldots,G_N(s)\}\), its DW shell is the convex hull of all the DW shells of the individual blocks \(G_{i}(j\omega)\). This block-diagonal structure admits the following decentralized stability conditions.
    \begin{lemma}[Geometric Decentralized Stability Condition]
    \label{thm:geometric-decentralized-stability}
    Consider $G$, $H \in \mathcal{RH}_{\infty}^{m\times m}$, where $G(s)$ is block-diagonal. The closed-loop system \( G(s)\#H(s) \) is stable if, for each \( \omega \in [0,\infty) \),
    \textbf{either}
    \begin{enumerate}[1)]
        \item \label{cond:gain}
        the decentralized \textbf{gain} condition, i.e.,
        \begin{equation}\label{eq:gain-condition}
        \max_i \, \sigma_{\max}(G_i(j\omega)) 
        < \sigma_{\min}(H^{-1}(j\omega))
        \  \quad{\rm \textit{holds, \textbf{or}}}
        \end{equation}
        
        \item \label{cond:phase}
        the decentralized \textbf{phase} condition, i.e.,
        \begin{equation}\label{eq:phase-condition}
        \begin{cases}
        \textit{a) } \max\limits_i \, \phi_{\max}(G_i(j\omega)) < \pi - \phi_{\max}(H(j\omega)), \\[2pt]
        \textit{b) } \min\limits_i \, \phi_{\min}(G_i(j\omega)) > -\pi - \phi_{\min}(H(j\omega)), \\[2pt]
        \textit{and c) } \max\limits_i \, \phi_{\max}(G_i(j\omega)) 
        - \min\limits_i \, \phi_{\min}(G_i(j\omega)) < \pi,
        \end{cases}
        \vspace{-2mm}
        \end{equation}
        \textit{holds, \textbf{or}}
        
        \item \label{cond:xz-separation}
        $H(j\omega)=H^{*}(j\omega)\succ0$ and the decentralized $x$-$z$
        \textbf{graph separation} condition holds, i.e., for each
        $i$,
        \begin{equation}
        \label{eq:xz-separation-condition}
        P(G_i(j\omega)) \cap
        P\!\left(-\tfrac{1}{\tau}H^{-1}(j\omega)\right)
        =\varnothing,
        \quad \forall \tau\in(0,1].
        \end{equation}

    \end{enumerate}
    \end{lemma}

    The proof of Lemma~\ref{thm:geometric-decentralized-stability} is similar to the proof of Theorem~4.7 in~\cite{huang2025DW}, while here we consider a more general case in~\eqref{eq:xz-separation-condition} where $H$ is not necessarily a constant matrix. Note that Lemma~\ref{thm:geometric-decentralized-stability} can already be applied to analyze the stability of $\widetilde{\bm{Y}}_{\text{CBR}}(s) \#  \widetilde{\bm{Y}}_{\rm Grid}^{-1}(s)$, since $\widetilde{\bm{Y}}_{\text{CBR}}(s)$ is block-diagonal and $\widetilde{\bm{Y}}_{\rm Grid}^{-1}(s)$ is positive definite. However, To quantify the stabilizing effect of heterogeneous GFM converters, we focus on the interconnection \(\widetilde{\bm Y}_{\rm GFL}(s)\#\widetilde{\bm Y}_{\rm GridC}^{-1}(s)\) in \eqref{eq:reduced-characteristic}. The next section characterizes how the GFM dynamics reshape the DW shell of \(\widetilde{\bm Y}_{\rm GridC}(s)\).


    \section{DW Shell of the Equivalent Network and Impact of GFM Dynamics} \label{sec:gfm-support-analysis}
 
    This section characterizes the equivalent network through its DW shell. We first analyze the real axis and imaginary axis projections of the DW shell. Then, the local passivity index and the imaginary-axis index of the GFM devices are defined to derive certified bounds on these projections. The resulting bounds constitute an envelope enabling the stability certificates that account for the heterogeneous GFM dynamics.
    
    \subsection{An Envelope for Equivalent Network's DW Shell}
   
     The exact DW shell of the equivalent network can be difficult to obtain when heterogeneous GFM converters are considered, as seen from the complicated computation of \(\widetilde{\bm Y}_{\mathrm{GridC}}(j\omega)\) in~\eqref{eq:equivalent network-admittance}. We therefore derive bounds of the DW shell to simplify the computation and analysis through the coordinate projections of DW shell.
     For a complex matrix \( A\), let \(\mathcal H( A):=(A+A^*)/2\) and \(\mathcal K(A):=(A- A^*)/(2j)\) denote its Hermitian part and its imaginary Hermitian part, respectively. 
     At a given frequency, the minimum value of \(DW(\widetilde{\bm Y}_{\mathrm{GridC}}(j\omega))\) in the real coordinate is determined by the Hermitian part of the equivalent network admittance
     \begin{equation}\label{eq:alpha}
         \alpha(j\omega):=
    \lambda_{\min} (\mathcal H(\widetilde{\bm Y}_{\mathrm{GridC}}(j\omega))),
     \end{equation}
     where \(\lambda_{\min}(\cdot)\) denotes the minimum eigenvalue of a Hermitian matrix. The quantity \(\alpha(j\omega)\) can be understood as the power grid strength, as it equals the generalized short-circuit ratio (gSCR) if $\widetilde{\bm Y}_{\mathrm{GridC}}(j\omega)$ is a constant matrix~\cite{huang2024gain,dong2019small}. 
     The real-axis values of every point in \(DW(\widetilde{\bm Y}_{\mathrm{GridC}}(j\omega))\) are no smaller than \(\alpha(j\omega)\).  
      Note that the equivalent network consists of the original ac network and the embedded GFM dynamics, where the ac network determines the underlying grid strength, while the GFM devices can equivalently increase the grid strength thanks to their voltage support. Hence, for the systems considered in this paper, the Hermitian part of \(\widetilde{\bm Y}_{\mathrm{GridC}}(j\omega)\) remains positive definite, and \(\alpha(j\omega)>0\). This property is certified by the local GFM bounds derived in the following subsection.
      The imaginary-axis projection of the DW shell determines the phase sector associated with $\alpha(j\omega)$. To be specific, the imaginary-axis values of every point in \(DW(\widetilde{\bm Y}_{\mathrm{GridC}}(j\omega))\) are within $[-\beta(j\omega),\beta(j\omega)]$, where
      \begin{equation}\label{eq:beta}
          \beta(j\omega):= \sigma_{\max}(\mathcal K(\widetilde{\bm Y}_{\mathrm{GridC}}(j\omega))).
      \end{equation}
      Note that $[-\beta(j\omega),\beta(j\omega)]$ also bounds the imaginary-axis values of the numerical range of $\widetilde{\bm Y}_{\mathrm{GridC}}(j\omega)$, which is sectorial and excludes the origin because \(\alpha(j\omega)>0\). 

      
      We consider $H^{-1}(j\omega) = \widetilde{\bm Y}_{\mathrm{GridC}}(j\omega)$ in~\eqref{eq:centralized-DW shell-separation} and thus for \(-\tfrac{1}{\tau}\widetilde{\bm Y}_{\mathrm{GridC}}(j\omega)\), the bounds become \(x\leq-\tfrac{1}{\tau}\alpha(j\omega)\) and \(|y|\leq \tfrac{1}{\tau}\beta(j\omega)\). It can be seen that their ratio is independent of \(\tau\), and thus the phase sector is $[-\theta(j\omega),\theta(j\omega)]$ on the left-half plane as shown in Fig.~\ref{fig:equivalent network-DW shell-enclosure}, where  
      \begin{equation}\label{eq:theta}
\theta(j\omega)=\tan^{-1}(\beta(j\omega)/\alpha(j\omega)) .
      \end{equation}
      We also notice that every point \((x,y,z)\) in a DW shell should satisfy \(z\geq x^2+y^2\).
      These bounds yield the following outer envelope of the equivalent network's  DW shell trajectory (over $\tau \in (0,1]$).
 
     \begin{lemma}
    [An Outer Envelope of the Equivalent Network's DW Shell]  \label{def:equivalent-network-dw-shell-enclosure}
     At a given frequency \(\omega\), the DW shell trajectory of \(-\tfrac{1}{\tau}\widetilde{\bm Y}_{\mathrm{GridC}}(j\omega)\), where \(\tau\in(0,1]\), is enclosed by
     \begin{equation}
     \begin{aligned}
     \mathcal R_{\alpha,\theta}(j\omega)=\{(x,y,z)\mid\;&z\ge x^2+y^2,\; x\le-\alpha(j\omega),\\& |y|\le -x\tan\theta(j\omega)\}. \end{aligned}
     \label{eq:equivalent network-DW shell-enclosure}
      \end{equation}
      Moreover, the \(x\)-\(y\) and \(x\)-\(z\) projections of this envelope are
      \begin{equation}
      \begin{aligned}
      \Pi_{xy}\mathcal R_{\alpha,\theta}(j\omega)
     =\bigl\{ & (x,y)\,\bigm|\,
    x\leq-\alpha(j\omega), \notag\\[-1mm]
      & \, |y|\leq-x\tan\theta(j\omega)\bigr\},
     \\
     \Pi_{xz}\mathcal R_{\alpha,\theta}(j\omega)
    =\bigl\{&(x,z)\,\bigm|\,
     x\leq-\alpha(j\omega),\;
    z\geq x^2\bigr\}.
     \end{aligned}
      \end{equation}
     \end{lemma}
     
     The geometric shape of the envelope $\mathcal{R}_{\alpha,\theta}(j\omega)$ is illustrated in Fig.~\ref{fig:equivalent network-DW shell-enclosure}, where Fig.~\ref{fig:equivalent network-DW shell-enclosure}~(a) shows an example DW shell and its outer envelope obtained by~\eqref{eq:equivalent network-DW shell-enclosure}. Fig.~\ref{fig:equivalent network-DW shell-enclosure}~(b) and~(c) show the corresponding \(x\)-\(z\) and \(x\)-\(y\) projections of the DW shell and its envelope, respectively. The \(x\)-\(z\) projection is determined by the equivalent power grid strength $\alpha(j\omega)$, while the \(x\)--\(y\) projection further incorporates the phase sector bounds. 

        \begin{figure}[!t]
    	\centering
        \vspace{-0mm}
    	\includegraphics[width=3.3in]{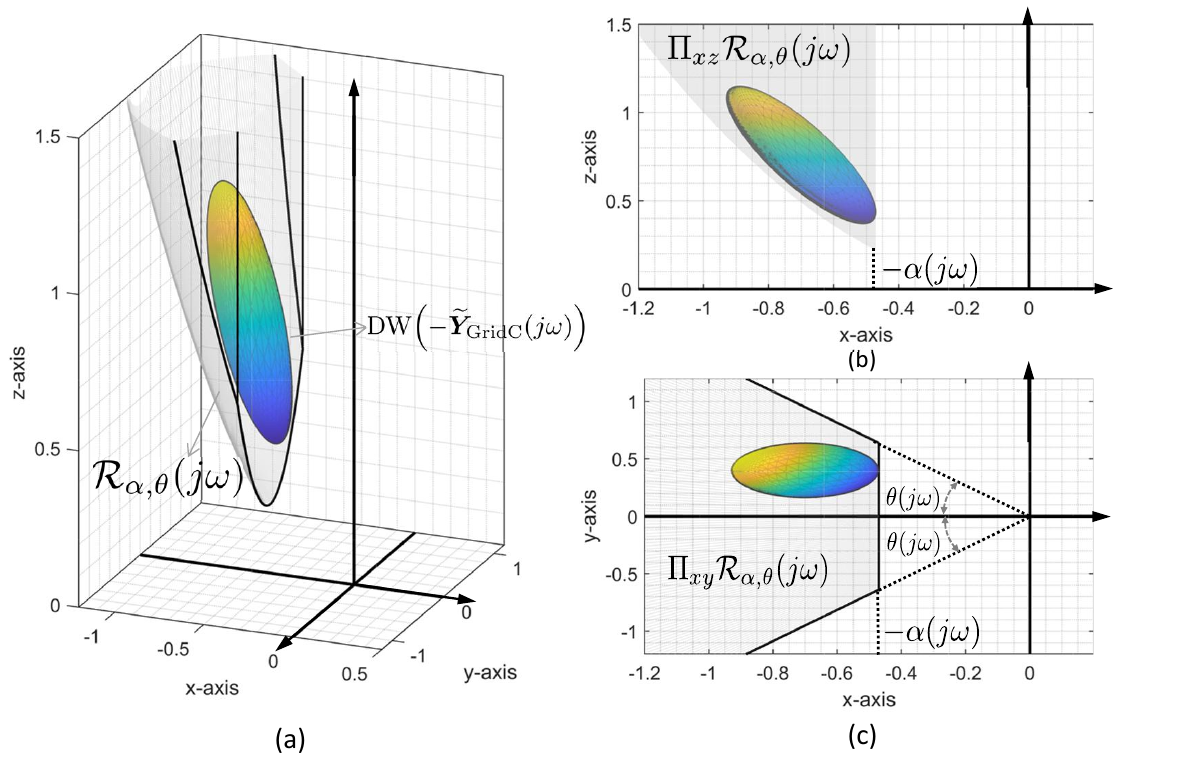}
    	\vspace{-3mm}
    	\caption{DW shell envelope $\mathcal{R}_{\alpha,\theta}(j\omega)$ of the equivalent network. (a) DW shell and its envelope. (b) $x$-$z$ projection. (c) $x$-$y$ projection.}
    	\vspace{-3mm}
	\label{fig:equivalent network-DW shell-enclosure}
      \end{figure}

    \vspace{-2mm}

     \subsection{GFM Indices and their Impacts on DW Shell Envelope}

     To capture the essential dynamics and stabilizing effects of GFM converters, define the \textbf{passivity index} $\nu_i(j\omega)$ and the \textbf{imaginary-axis index} $\mu_i(j\omega)$ of the $i$-th GFM converter as
     \begin{equation} \label{eq:gfm-local-indices}
      \begin{aligned}
       \nu_i(j\omega) &= \lambda_{\min}( \mathcal H( \widetilde{\bm Y}_{\mathrm{GFM},i}(j\omega) )),\\ 
       \mu_i(j\omega) &= \sigma_{\max}( \mathcal K( \widetilde{\bm Y}_{\mathrm{GFM},i}(j\omega) )).
     \end{aligned}
     \end{equation}
     For the $m$ GFM converters, the indices are extended to
     \begin{equation}
     \begin{aligned}
\bm\Gamma(j\omega)&=\operatorname{diag}\left(\nu_1(j\omega),\ldots,\nu_m(j\omega)\right), \\
     \bm\Xi(j\omega) &= \operatorname{diag}\left( \mu_1(j\omega),\ldots,\mu_m(j\omega) \right).
     \label{eq:gfm-local-index-matrices}
     \end{aligned}
    \end{equation}

    The above two indices can be conveniently obtained from the black-box admittance models of GFM converters (e.g., via frequency scanning), without knowing the detailed control scheme and parameters. Moreover, they compress the original $2 \times 2$ transfer function matrix by extracting the key information of how the GFM dynamics will affect the equivalent network described by $\widetilde{\bm Y}_{\mathrm{GridC}}(j\omega)$, as shown in the following result.



    \begin{theorem}[From GFM indices to DW Shell Envelope] \label{prop:equivalent-network-bounds}
    Consider the GFM indices $\bm\Gamma(j\omega)$ and $\bm\Xi(j\omega)$ as in~\eqref{eq:gfm-local-index-matrices}. At a frequency $\omega \in [0,\infty)$ that satisfies $\bm Q_D+\bm\Gamma(j\omega)\succ0$, we have   
    \begin{equation}
    \alpha(j\omega) \ge \alpha_{\mathrm{LB}}(j\omega) := \lambda_{\min}[ \bm Q_A - \bm Q_B \left( \bm Q_D+\bm\Gamma(j\omega) \right)^{-1} \bm Q_C ],
    \label{eq:equivalent-grid-strength-lower-bound}
     \end{equation}
    \begin{equation}
     \begin{aligned}
     \beta(j\omega) \le \beta_{\rm UB}(j\omega):=&\;\|\bm Q_B\bm Q_D^{-\frac{1}{2}}\|\|\bm Q_D^{-\frac{1}{2}}\bm Q_C\|\\
     &\times\frac{\lambda_{\max}\left(\bm Q_D^{-\frac{1}{2}}\bm\Xi(j\omega)\bm Q_D^{-\frac{1}{2}}\right)}{\left[1+\lambda_{\min}\left(\bm Q_D^{-\frac{1}{2}}\bm\Gamma(j\omega)\bm Q_D^{-\frac{1}{2}}\right)\right]^2,}
     \end{aligned}\label{eq:equivalent network-imaginary-axis-upper-bound}
     \end{equation}
     where $\alpha(j\omega)$ and $\beta(j\omega)$ are the tight bounds of the DW shell of $\widetilde{\bm Y}_{\mathrm{GridC}}(j\omega)$, as given in~\eqref{eq:alpha} and~\eqref{eq:beta}. We further have
     \begin{equation} 
     \theta(j\omega) \le \theta_{\rm UB}(j\omega):=\tan^{-1}\frac{\beta_{\rm UB}(j\omega)}{\alpha_{\rm LB}(j\omega)}\,,
      \label{eq:equivalent network-phase-sector-upper-bound}
     \end{equation}
     where $\theta(j\omega)$ is the phase bound in~\eqref{eq:theta}, and ${\alpha_{\rm LB}(j\omega)}>0$ because of the stabilizing effect of GFM converters.
     \end{theorem}

     \begin{proof}
     At a frequency $\omega \in [0,\infty)$, we omit the argument of $j\omega$ in the transfer function (matrix) for simplicity and define
     \[\bm M=I_{2m}+(\bm Q_D^{-\frac{1}{2}}\otimes I_2)\widetilde{\bm Y}_{\rm GFM}(\bm Q_D^{-\frac{1}{2}}\otimes I_2).\]
     
     From the definitions of $\bm\Gamma$ and $\bm\Xi$, we have
     \[\mathcal H(\bm M)\succeq \left(I_m+\bm Q_D^{-\frac{1}{2}}\bm\Gamma\bm Q_D^{-\frac{1}{2}}\right) \otimes I_2\succ0 , {\rm and} \]
     $$
     \|\mathcal K(\bm M)\|\le \lambda_{\max}\!\left(\bm Q_D^{-\frac{1}{2}}\bm\Xi\bm Q_D^{-\frac{1}{2}}\right).
     $$
     Here, \(A\succeq B\) means that \(A-B\) is positive semi-definite.
     
     By further considering that \(\bm S_{\rm GFM}=\bm M^{-1}\) and \(\mathcal H(\bm M^{-1})\preceq\mathcal H(\bm M)^{-1}\), we obtain \(\mathcal H(\bm S_{\rm GFM})\preceq (I_m+\bm Q_D^{-\frac{1}{2}}\bm\Gamma \bm Q_D^{-\frac{1}{2}})^{-1}\otimes I_2\). By substituting the above equation into \eqref{eq:equivalent network-admittance} and considering \(\bm Q_D^{-\frac{1}{2}}(I_m+\bm Q_D^{-\frac{1}{2}}\bm\Gamma \bm Q_D^{-\frac{1}{2}})^{-1}\bm Q_D^{-\frac{1}{2}} =(\bm Q_D+\bm\Gamma)^{-1}\), we derive
     $$
     \mathcal H(\widetilde{\bm Y}_{\rm GridC})\succeq \left[\bm Q_A-\bm Q_B(\bm Q_D+\bm\Gamma)^{-1}\bm Q_C\right] \otimes I_2,
     $$
     which proves \eqref{eq:equivalent-grid-strength-lower-bound}. Since \(\sigma_{\min}(\bm M)\geq 1+\lambda_{\min}(\bm Q_D^{-\frac{1}{2}}\bm\Gamma \bm Q_D^{-\frac{1}{2}})\) and \(\mathcal K(\bm M^{-1})=-(\bm M^*)^{-1} \mathcal K(\bm M)\bm M^{-1}\), it follows that
     $$
     \left\|\mathcal K(\bm S_{\rm GFM})\right\|_2\leq \frac{\lambda_{\max}\!\left(\bm Q_D^{-\frac{1}{2}}\bm\Xi\bm Q_D^{-\frac{1}{2}}\right)}{\left[1+\lambda_{\min}\!\left(\bm Q_D^{-\frac{1}{2}}\bm\Gamma\bm Q_D^{-\frac{1}{2}}\right)\right]^2.}
     $$
     Substituting this bound into \eqref{eq:equivalent network-admittance} and applying the submultiplicative property of the two-norm operation then proves the bound on \(\beta(j\omega)\). The expression of \(\theta_{\rm UB}(j\omega)\) then follows from the inequalities of \(\alpha_{\rm LB}(j\omega)\) and \(\beta_{\rm UB}(j\omega)\).
     \end{proof}

    The results in Theorem~\ref{prop:equivalent-network-bounds} explicitly demonstrate how the passivity index $\nu_i(j\omega)$ and the imaginary-axis index $\mu_i(j\omega)$ of GFM converters affect the DW shell envelope of the equivalent network. For instance, it can be seen from~\eqref{eq:equivalent-grid-strength-lower-bound} that the passivity index $\nu_i(j\omega)$ mainly affects $\alpha(j\omega)$, which determines the right bound in the $x$-$z$ graph; see Fig.~\ref{fig:equivalent network-DW shell-enclosure}~(b). This bound is closely related to the equivalent power grid strength, as investigated in~\cite{huang2025DW} via $x$-$z$ graph analysis. However, the impact of GFM converters was not considered in~\cite{huang2025DW}, and here we derive~\eqref{eq:equivalent-grid-strength-lower-bound} to theoretically show that it is the passivity index $\nu_i(j\omega)$ that decides whether a GFM converter, under certain control schemes and parameters, can increase the power grid strength or not. Moreover, since $\bm\Gamma(j\omega)$ enters~\eqref{eq:equivalent-grid-strength-lower-bound} through Kron reduction, we can deduce that a larger $\nu_i(j\omega)$ helps increase the power grid strength more. This also provides a guideline of designing GFM control: the control scheme and parameters should be chosen to increase $\nu_i(j\omega)$ as far as possible. The imaginary-axis index $\mu_i(j\omega)$ mainly affects the bound $\beta(j\omega)$ and the phase bound $\theta(j\omega)$, which describe the shape of the DW shell along the imaginary axis. It can be seen from~\eqref{eq:equivalent network-imaginary-axis-upper-bound} and~\eqref{eq:equivalent network-phase-sector-upper-bound} that a larger $\mu_i(j\omega)$ will likely result in a larger $\theta(j\omega)$, which, as will be shown below, may cause violation of the small-phase condition and result in instabilities. Hence, it is favorable to design the GFM control scheme and parameters to reduce $\mu_i(j\omega)$. 


    \vspace{-3mm}
    \subsection{Stability Certificates Based on the DW Shell Envelope}
     

     We use $\alpha_{\mathrm{LB}}(j\omega)$ and $\theta_{\mathrm{UB}}(j\omega)$ to construct the DW shell envelope \(\mathcal R_{\alpha_{\mathrm{LB}},\theta_{\mathrm{UB}}}(j\omega)\) according to~\eqref{eq:equivalent network-DW shell-enclosure}, which enables the following stability certificates to analyze the interaction between GFM and GFL converters.

     \begin{corollary}[Interaction Between GFL Converters and the Equivalent Network]
     \label{cor:centralized-DW shell-separation}
     If \(\widetilde{\bm Y}^{-1}_{\mathrm{GridC}}(s)\) and $\widetilde{\bm{Y}}_{\mathrm{GFL}}(s)$ are stable in open loop, then the closed-loop system $\widetilde{\bm Y}_{\rm GFL}(s)\# \widetilde{\bm Y}_{\rm GridC}^{-1}(s)$ which represents the interconnection in Fig.~\ref{fig:dynamic-repartitioning}, is stable if, for each \(\omega\in[0,\infty)\),
     \begin{equation}
    \operatorname{conv}
     \left(
     \bigcup_{i=1}^{n}
     DW(
     \widetilde{\bm Y}_{{\rm GFL},i}(j\omega)
    )
    \right)
    \cap
    \mathcal R_{\alpha_{\mathrm{LB}},\theta_{\mathrm{UB}}}(j\omega) =
     \varnothing.
    \label{eq:DW shell-separation}
    \end{equation}
    where \(\operatorname{conv}(\cdot)\) denotes the convex hull.
    \end{corollary}
     \begin{proof}
    As \(\widetilde{\bm Y}_{\rm GFL}(s)\) is block-diagonal, its DW shell is the convex hull of all $DW(\widetilde{\bm Y}_{{\rm GFL},i}(j\omega))$. Moreover, \(\mathcal R_{\alpha_{\mathrm{LB}},\theta_{\mathrm{UB}}}(j\omega)\) contains \(DW(-\frac{1}{\tau}\widetilde{\bm Y}_{\rm GridC}(j\omega))\) for all \(\tau\in(0,1]\). Let \(G(s)=\widetilde{\bm Y}_{\rm GFL}(s)\) and \(H(s)=\widetilde{\bm Y}_{\rm GridC}^{-1}(s)\), and one can deduce that \eqref{eq:DW shell-separation} ensures the DW shell separation required by Lemma~\ref{thm:geometric-centralized-stability}. 
    \end{proof}


    Then, we derive the corresponding decentralized stability certificates to enable the analysis of large-scale systems, which are based on projecting the DW shell and its envelope onto the $x$-$z$ and $x$-$y$ planes.

     \begin{corollary}[Decentralized Stability Certificates With an Equivalent Network Fusing GFM Dynamics]
      \label{cor:projection-based-decentralized-conditions}
      If $\widetilde{\bm{Y}}_{\mathrm{GFL}}(s)$ and \(\widetilde{\bm Y}_{\mathrm{GridC}}^{-1}(s)\) are stable in open loop, then the closed-loop system $\widetilde{\bm Y}_{\rm GFL}(s)\# \widetilde{\bm Y}_{\rm GridC}^{-1}(s)$ which represents the interconnection in Fig.~\ref{fig:dynamic-repartitioning}, is stable if,
      for each \(\omega\in[0,\infty)\), either

\begin{enumerate}[1)]
    \item \label{cond:equivalent network-xz-separation}
    the decentralized $x$-$z$ \textbf{graph separation} condition holds, i.e.,
    for each $i\in\{1,2,\ldots,n\}$,
    \begin{equation}
    \mathcal P(
    \widetilde{\bm Y}_{{\rm GFL},i}(j\omega)
    )
    \cap
    \Pi_{xz}\mathcal R_{\alpha_{\mathrm{LB}},\theta_{\mathrm{UB}}}(j\omega)
    =
    \varnothing , \quad {\rm or}
    \label{eq:equivalent network-xz-separation}
    \end{equation}
    \item \label{cond:equivalent network-small-phase}
    each \(\widetilde{\bm Y}_{\mathrm{GFL},i}(j\omega)\) is sectorial
    and the decentralized \textbf{small-phase} condition holds, i.e.,
 \begin{equation}
\resizebox{\columnwidth}{!}{$
\hspace{-4mm}\begin{cases}
\textit{a) } 
\max\limits_i\,\phi_{\max}(
\widetilde{\bm Y}_{{\rm GFL},i}(j\omega)
)<\pi-\theta_{\rm UB}(j\omega),\\[2pt]
\textit{b) } 
\min\limits_i\,\phi_{\min}(
\widetilde{\bm Y}_{{\rm GFL},i}(j\omega)
)>-\pi+\theta_{\rm UB}(j\omega),\\[2pt]
\textit{c) } 
\max\limits_i\,\phi_{\max}(
\widetilde{\bm Y}_{{\rm GFL},i}(j\omega)
)
-
\min\limits_i\,\phi_{\min}(
\widetilde{\bm Y}_{{\rm GFL},i}(j\omega)
)<\pi .
\end{cases}
$}
\label{eq:equivalent network-small-phase-condition}
\end{equation}
\end{enumerate}
\end{corollary}

     \begin{proof}
     In Condition~1), \(\Pi_{xz}\mathcal R_{\alpha_{\mathrm{LB}},\theta_{\mathrm{UB}}} (j\omega)\) is defined by \(x\leq-\alpha_{\mathrm{LB}}(j\omega)\) and \(z\geq x^2\). Since the GFL converter's \(x\)-\(z\) graph also satisfies \(z\geq x^2\), its separation from the envelope indicates \(x>-\alpha_{\mathrm{LB}}(j\omega)\), and thus the convex hull of all the GFL converters' \(x\)-\(z\) graph remains separated from the network's envelope.
     In Condition~2), the sub-condition~c) ensures that the numerical ranges of all GFL converters are within a common phase interval with a width smaller than \(\pi\). The sub-conditions~a) and b) ensure that this interval is separated from \(\Pi_{xy}\mathcal R_{\alpha_{\mathrm{LB}},\theta_{\mathrm{UB}}} (j\omega)\), and thus the convex hull of all the GFL converters' numerical ranges remains separated from the network's envelope. In short, both conditions imply the DW shell separation required by Corollary~\ref{cor:centralized-DW shell-separation} through suitable projections. This completes the proof.
\end{proof}

\vspace{-1mm}
     Based on the condition in~\eqref{eq:equivalent network-xz-separation}, the frequency-wise stability distance of the $i$-th GFL converter is defined by
      \begin{equation}
       \label{eq:xz-graph-stability-margin}
       \rho_i(j\omega):= \operatorname{dist}\left( P\bigl( \widetilde{\bm Y}_{\mathrm{GFL},i}(j\omega) \bigr), \Pi_{xz}\mathcal R_{\alpha_{\mathrm{LB}},\theta_{\mathrm{UB}}}(j\omega) \right),
      \end{equation}
      where \(\operatorname{dist}(\cdot,\cdot)\) denotes the minimum Euclidean distance between two sets. A positive \(\rho_i(j\omega)\) certifies the \(x\)-\(z\) graph separation at \(\omega\), while a zero margin indicates instability risk. This margin will be used to show how the individual GFL converters interact with the equivalent network and with the GFM converters across different frequencies.
\vspace{-2mm}
     \section{Illustrative Examples}\label{sec:case-studies}
     
     The proposed method is tested on a three-converter system and a modified IEEE 68-bus system. The three-converter system is used to illustrate the equivalent network reformulation and show how the GFM dynamics affect the system stability. The 68-bus system demonstrates the applicability of the method to large systems containing heterogeneous grid-forming devices. In both cases, the frequency-domain results are accompanied by time-domain simulations.
    \vspace{-2mm}
     \subsection{Case Study of a Three-Converter System}

    \begin{figure}[!t]
    	\centering
        \vspace{0.5mm}
    	\includegraphics[width=3.0in]{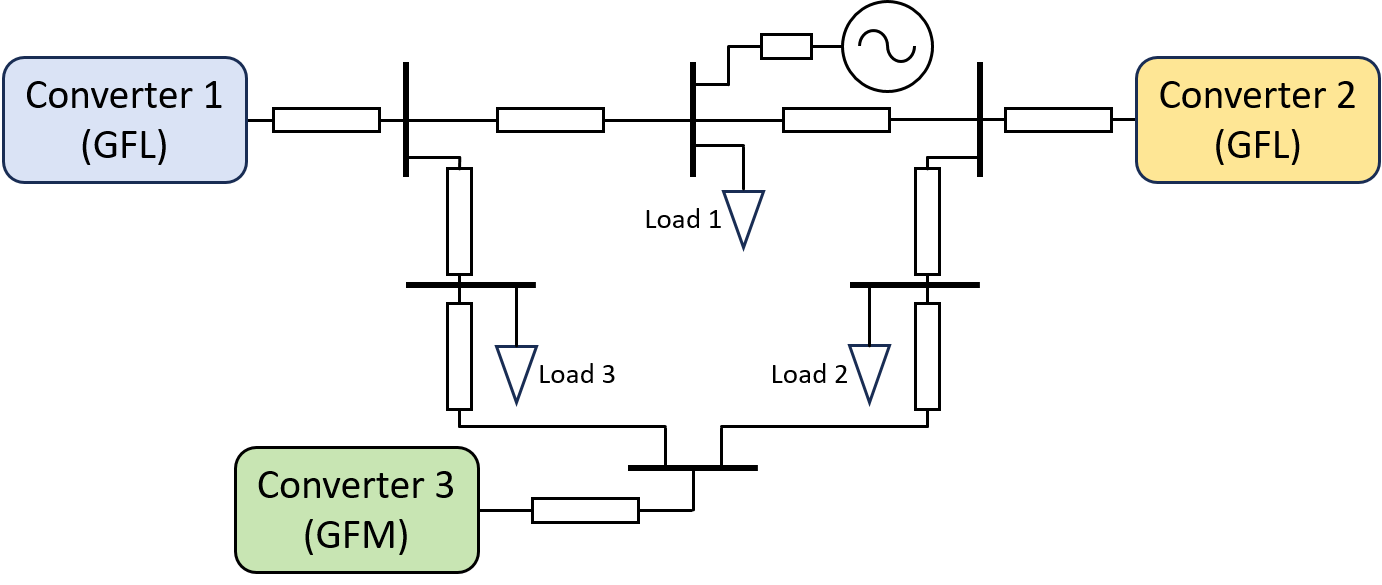}
    	\vspace{-2mm}
    	\caption{A three-converter test system.}
    	\vspace{-2mm}
	\label{fig:three-converter-system}
    \end{figure}
    
    Consider the three-converter system shown in Fig.~\ref{fig:three-converter-system}, where Converters 1 and 2 operate in GFL mode and Converter 3 operates in GFM mode. The PLL bandwidths of Converters 1 and 2 are $80$~rad/s and $60$~rad/s, respectively. The other parameters are provided in the supplementary material. We aim to show that when analyzing a system with both GFM and GFL converters, the GFM dynamics should be fused into the network so that one can see how GFM converters equivalently increase the power grid strength. 

    Fig.~\ref{fig:three-converter-standard-certificate} first shows the frequency-domain plots when the GFM dynamics are not fused into the network, that is, the formulation in~\eqref{eq:rescaled-interconnection} is used for the decentralized stability analysis enabled by Lemma~\ref{thm:geometric-decentralized-stability}. The displayed distances are computed based on~\eqref{eq:xz-separation-condition}, i.e., the distance between the converter's $x$-$z$ graph and the network's $x$-$z$ graph. If the distance is larger than 0, then the $x$-$z$ graph separation condition~\eqref{eq:xz-separation-condition} is satisfied. 
    Similar to~\cite{huang2024gain}, the phase area of the network is defined as $[-\pi - \phi_{\min}(\widetilde{\bm{Y}}_{\rm Grid}(j\omega)),\pi - \phi_{\max}(\widetilde{\bm{Y}}_{\rm Grid}(j\omega))]$ and the phase area of the $i$-th converter is $[\phi_{\min}(\widetilde{\bm Y}_{{\rm CBR},i}(j\omega)),\phi_{\max}(\widetilde{\bm Y}_{{\rm CBR},i}(j\omega))]$. Hence, if the converter's phase area is contained in the network's phase area, then the phase condition~\eqref{eq:phase-condition} is satisfied. It can be seen from Fig.~\ref{fig:three-converter-standard-certificate} that the phase condition is satisfied for all the converters above 42~Hz. While below 42~Hz, Converter~1's distance is zero between 9~Hz and 18~Hz, indicating violations of the $x$-$z$ graph separation condition~\eqref{eq:xz-separation-condition} in this frequency range. Hence, the system cannot be certified to be stable based on the analysis in Fig.~\ref{fig:three-converter-standard-certificate} due to conservativeness of the analysis.

    \begin{figure}[!t]
    	\centering
        \vspace{0.5mm}
    	\includegraphics[width=3.0in]{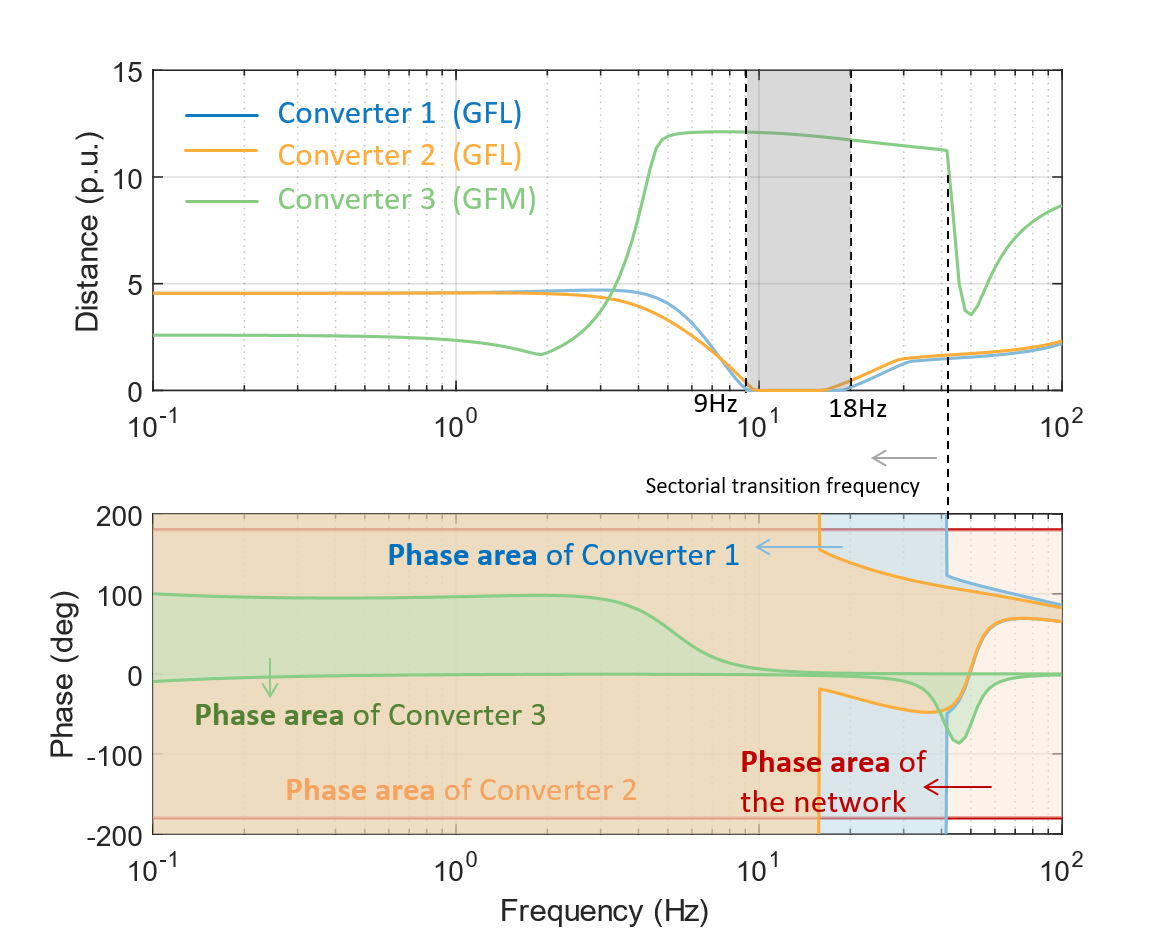}
    	\vspace{-4mm}
    	\caption{Stability distances and phase areas when GFM dynamics are not fused into the network.}
    	\vspace{-3mm}
	\label{fig:three-converter-standard-certificate}
    \end{figure}

    \begin{figure}[!t]
    	\centering
        \vspace{-2mm}
    	\includegraphics[width=3.0in]{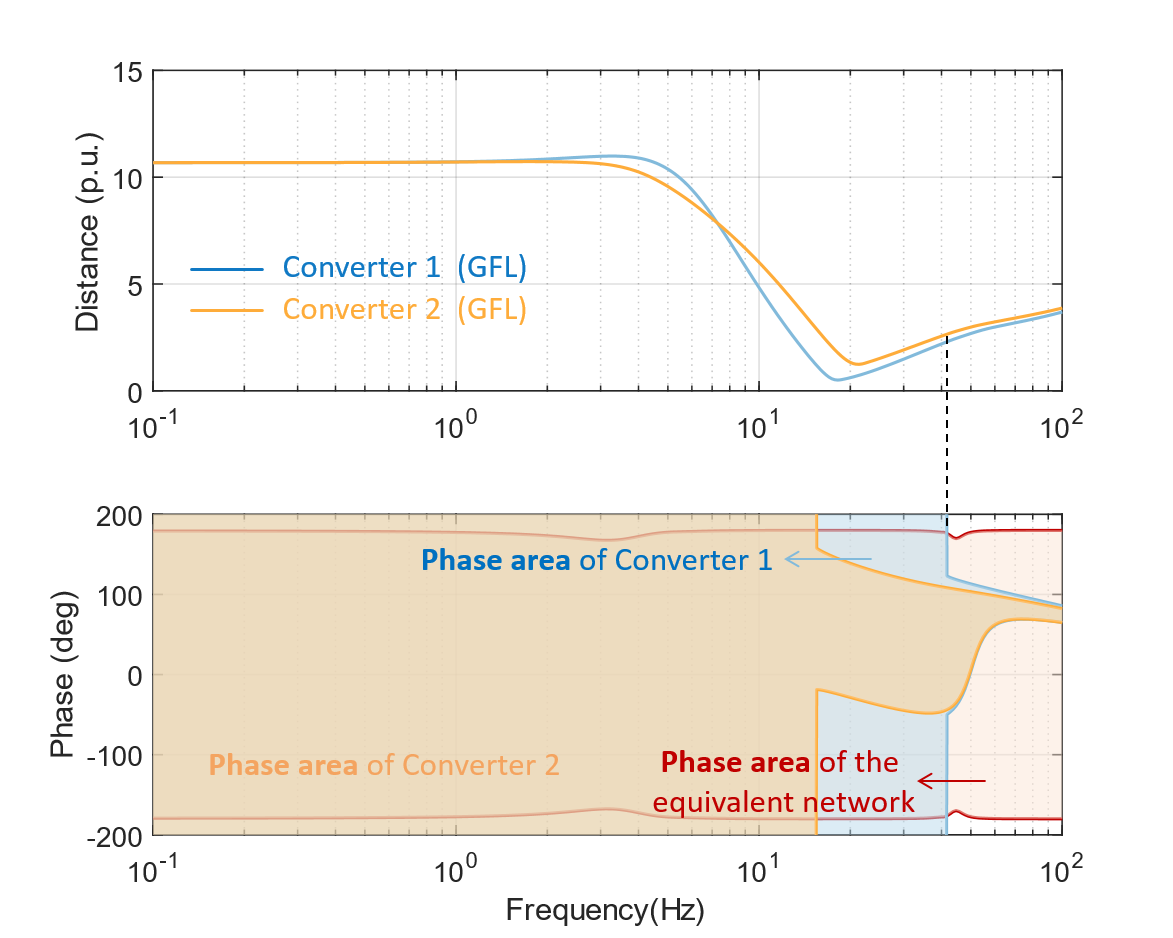}
    	\vspace{-4mm}
    	\caption{Stability distances and phase areas when GFM dynamics are fused into the network.}
    	\vspace{-3mm}
	\label{fig:three-converter-equivalent network-certificate}
    \end{figure}

    We then follow the procedure in~\eqref{eq:schur-complement-factorization} and~\eqref{eq:equivalent network-admittance} to fuse the GFM dynamics into the network and obtain the equivalent network $\widetilde{\bm{Y}}_{\mathrm{GridC}}(s)$, which can significantly reduce the conservativeness since it reflects how GFM converters enhance the power grid strength. The passivity index and the imaginary-axis index of the GFM converter are computed according to~\eqref{eq:gfm-local-indices}, and then they are used to compute $\alpha_{\mathrm{LB}}(j\omega)$ and $\theta_{\mathrm{UB}}(j\omega)$ according to Theorem~\ref{prop:equivalent-network-bounds} and construct the DW shell envelope \(\mathcal R_{\alpha_{\mathrm{LB}},\theta_{\mathrm{UB}}}(j\omega)\). On this basis, Fig.~\ref{fig:three-converter-equivalent network-certificate} plots the distance between the GFL converters' $x$-$z$ graphs and the $x$-$z$ graph of the equivalent network, i.e., $\rho_i(j\omega)$ in~\eqref{eq:xz-graph-stability-margin}, as well as the phase area of the equivalent network defined by $[-\pi+\theta_{\rm UB}(j\omega), \pi - \theta_{\rm UB}(j\omega)]$. It can be seen that the distances of the two GFL converters are larger than 0 across the whole frequency range, indicating that the system is stable. 

    \begin{figure}[!t]
    	\centering
        \vspace{0.5mm}
    	\includegraphics[width=2.8in]{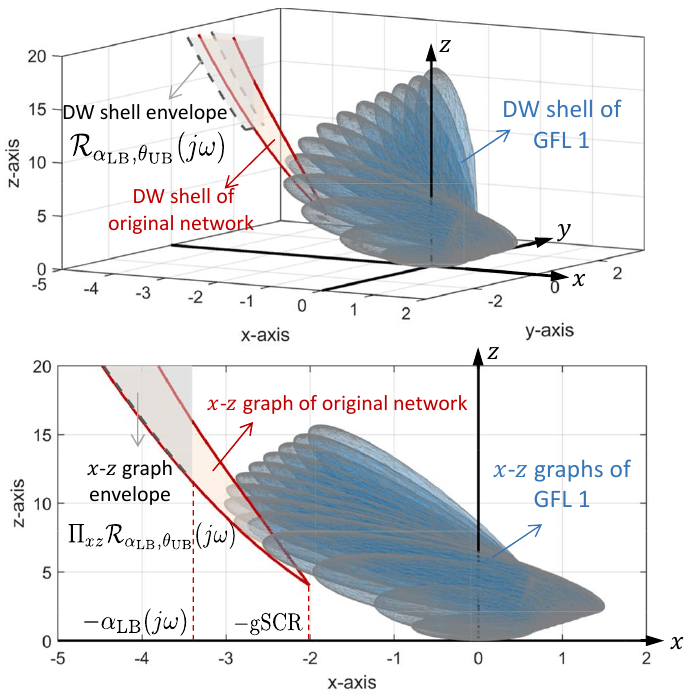}
    	\vspace{-4.5mm}
    	\caption{DW shell envelope and $x$-$z$ graph envelope of the original network and the equivalent network from 5~Hz to 30~Hz.}
    	\vspace{-3mm}
	\label{fig:three-converter-xz-graphs}
    \end{figure}

    \begin{figure}[!t]
    	\centering
        \vspace{0.5mm}
    	\includegraphics[width=2.6in]{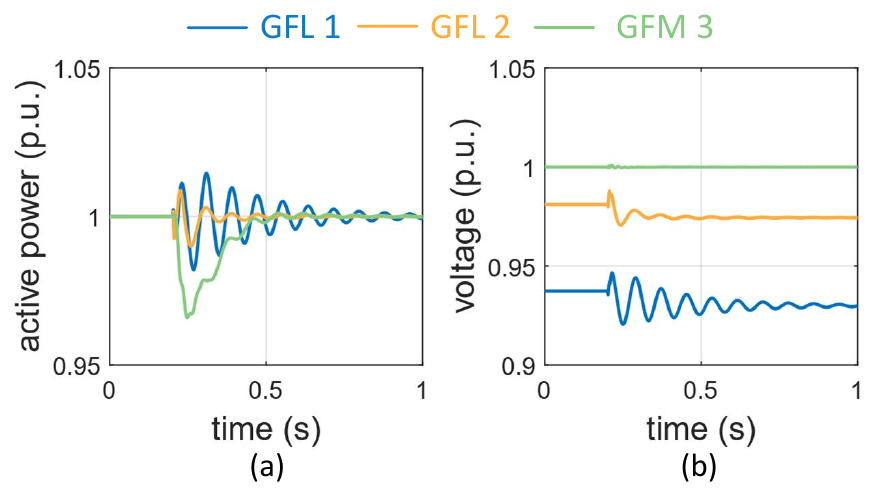}
    	\vspace{-4.5mm}
    	\caption{Time-domain responses of the three-converter system: (a) active power, and (b) voltage magnitude.}
    	\vspace{-4mm}
	\label{fig:three-converter-time-domain}
    \end{figure}
    
    We further plot the DW shells and $x$-$z$ graphs in Fig.~\ref{fig:three-converter-xz-graphs} to show why fusing the GFM dynamics into the network can reduce the conservativeness. 
    It can be seen that the original DW shell and $x$-$z$ graph of the network both intersect with the GFL Converter~1 due to a low power grid strength characterized by a low gSCR, aligned with the results in Fig.~\ref{fig:three-converter-standard-certificate}. After fusing the GFM dynamics into the network, the DW shell envelope \(\mathcal R_{\alpha_{\mathrm{LB}},\theta_{\mathrm{UB}}}(j\omega)\) and the $x$-$z$ graph envelope $\Pi_{xz}\mathcal R_{\alpha_{\mathrm{LB}},\theta_{\mathrm{UB}}}(j\omega)$ of the equivalent network become separated from the GFL Converter~1, which indicate that the system is stable, aligned with the results in~Fig.~\ref{fig:three-converter-equivalent network-certificate}. We can see that the power grid strength increases from gSCR to at least $\alpha_{\rm LB}(j\omega)$ by fusing the GFM dynamics into the network, which reflects how GFM converters improve the stability of GFL converters and justifies the necessity of deriving the equivalent network. 
    By comparison, the results in Fig.~\ref{fig:three-converter-standard-certificate} are conservative because it does not consider the interaction between GFM and GFL converters. 
    Fig.~\ref{fig:three-converter-time-domain} plots the time-domain responses of the system, where a small disturbance occurs at $t=0.2$~s. It can be seen that the system is stable, fully aligned with the results of the equivalent network analysis in Fig.~\ref{fig:three-converter-equivalent network-certificate} and Fig.~\ref{fig:three-converter-xz-graphs}.

    \begin{figure}[!t]
    	\centering
        \vspace{0mm}
    	\includegraphics[width=2.7in]{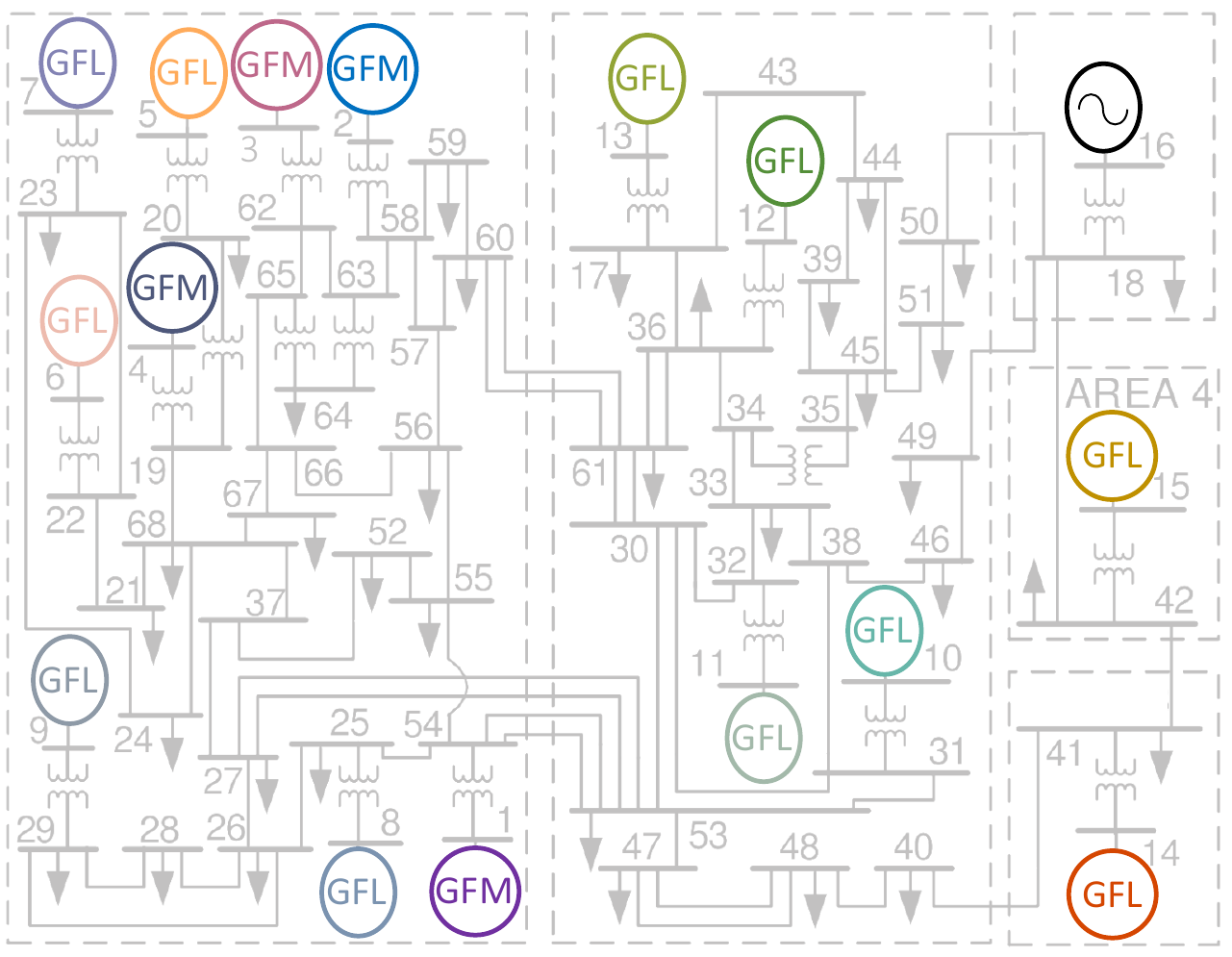}
    	\vspace{-3mm}
    	\caption{A 68-bus system with 4 GFM converters and 11 GFL converters.}
    	\vspace{-1mm}
	\label{fig:modified-ieee-68bus-system}
    \end{figure}

    \subsection{Case Study of the 68-Bus System}

     We use a modified IEEE 68-bus system to show the effectiveness of the proposed method when handling heterogeneous GFM devices. As shown in Fig.~\ref{fig:modified-ieee-68bus-system}, four GFM converters are connected to Buses 1-4, and eleven GFL converters are connected to Buses 5-15. Bus 16 is an infinite bus to model a remote area weakly connected to the system. The detailed system parameters are provided in the supplementary material.

     Two cases are considered: in Case~1, GFM Converters 1 and 2 adopt VSG control, GFM Converter 3 adopts droop control, and GFM Converter 4 adopts matching control to handle dc-link dynamics~\cite{Huang2017ViSynC}; in Case~2, GFM Converter~1 is replaced by a synchronous generator (SG) under the same capacity, while all other components and control settings remain unchanged.
     Fig.~\ref{fig:gfm-passivity-indices} compares the passivity indices of the GFM converters and the SG, which show quite different passivity profiles. For instance, the passivity index of the SG is higher than those of the other devices within $[5~{\rm Hz},~12~{\rm Hz}]$, indicating a higher stabilizing effect in this frequency range.


\begin{figure}[!t]
    \centering
    \includegraphics[width=2.7in]{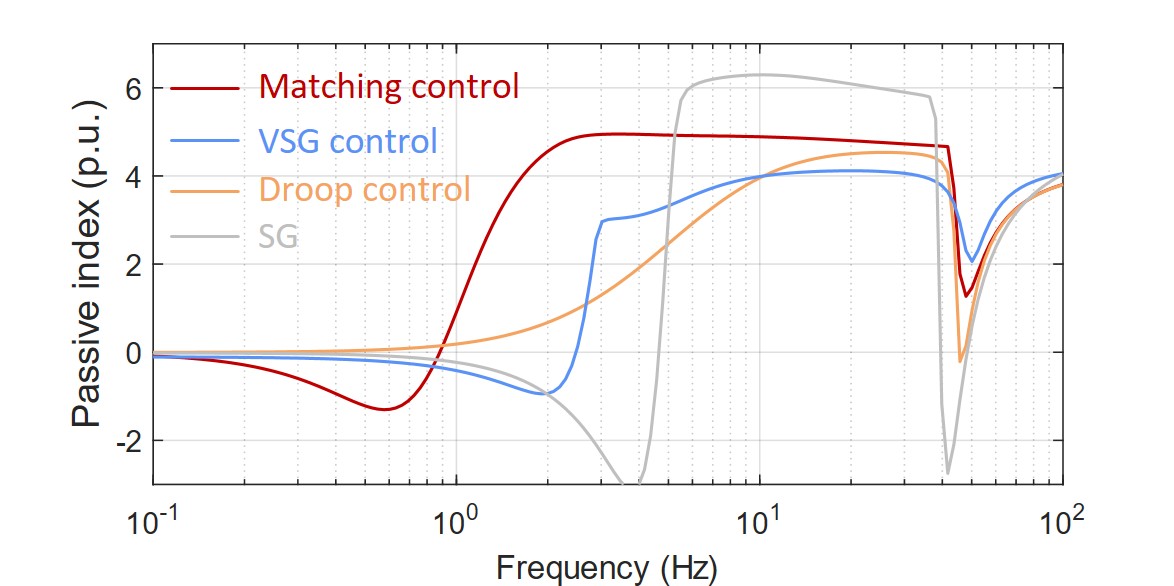}
    \vspace{-3mm}
    \caption{Passivity indices of the four different GFM devices.}
    \label{fig:gfm-passivity-indices}
\end{figure}

    Consider the setting of Case~1. Fig.~\ref{fig:ieee-68bus-frequency-certificate} plots the stability distances of all the GFL converters, which are computed according to~\eqref{eq:xz-graph-stability-margin}, as well as the phase areas of the GFL converters and the equivalent network. Notice that the computations of the stability distances involve the DW shell envelope of the equivalent network, which is obtained from the two GFM indices using Theorem~\ref{prop:equivalent-network-bounds}. The phase area of the equivalent network is still defined by $[-\pi+\theta_{\rm UB}(j\omega), \pi - \theta_{\rm UB}(j\omega)]$. It can be seen from Fig.~\ref{fig:ieee-68bus-frequency-certificate} that GFL~5 has the largest sectorial transition frequency, which is 16~Hz. Above 16~Hz, the small-phase condition is satisfied. However, below 16~Hz, there exists a frequency range where the stability distances of GFL~5 and GFL~6 are zero, that is, condition~\eqref{eq:equivalent network-xz-separation} is violated, indicating an instability risk.


\begin{figure}[!t]
    \centering
    \includegraphics[width=3in]{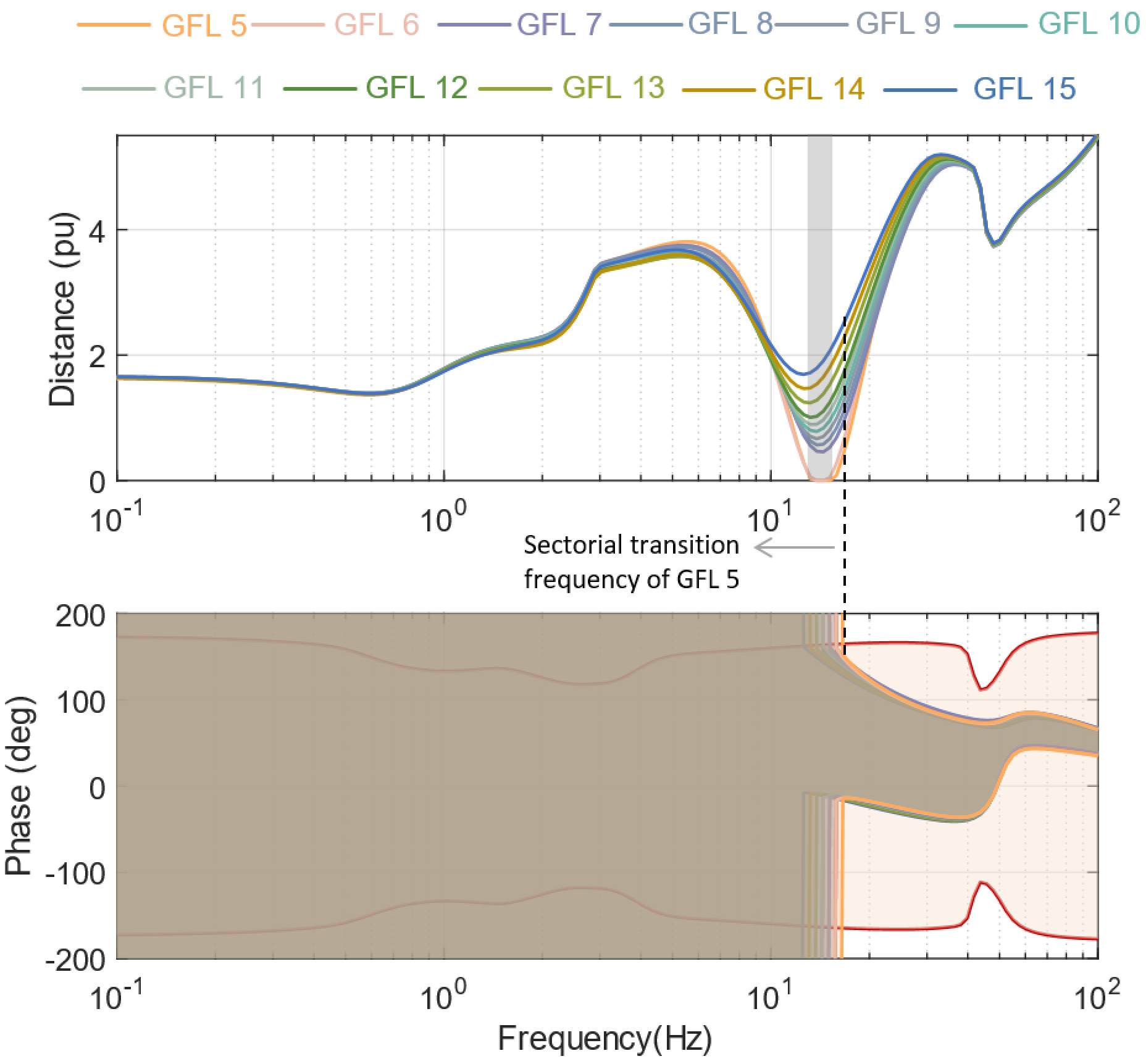}
    \vspace{-3mm}
    \caption{Frequency-wise stability distances and phase areas of the GFL converters and the equivalent network in Case~1.}
    \vspace{-3mm}
    \label{fig:ieee-68bus-frequency-certificate}
\end{figure}

    Fig.~\ref{fig:ieee-68bus-critical-frequency-range}~(a) is the zoomed-in version of Fig.~\ref{fig:ieee-68bus-frequency-certificate} around the critical frequency range where the stability distances approach zero. Fig.~\ref{fig:ieee-68bus-critical-frequency-range}~(b) further provides the corresponding results of Case~2, where GFM Converter~1 in Case~1 is replaced by an SG with the same capacity. It can be seen from Fig.~\ref{fig:ieee-68bus-critical-frequency-range}~(b) that the stability distances of GFL~5 and GFL~6 become larger than 0, and thus the system is guaranteed to be stable. Such a stability improvement compared to Fig.~\ref{fig:ieee-68bus-critical-frequency-range}~(a) is attributed to the higher passivity index of the SGs in the critical frequency range, as shown in Fig.~\ref{fig:gfm-passivity-indices}. To further demonstrate why the stability distances become larger than zero in Fig.~\ref{fig:ieee-68bus-critical-frequency-range}~(b), we plot in Fig.~\ref{fig:Dwshell15Hz} the DW shell envelope of the equivalent network and the convex hull of the GFL converters' DW shells at 14~Hz, for both Case~1 and Case~2. It can be seen that the DW shell envelope of the equivalent network in Case~1 intersects with the converter's DW shell, aligned with the zero stability distance in Fig.~\ref{fig:ieee-68bus-critical-frequency-range}~(a). Thanks to the higher passivity index of the SG at 14~Hz, Case~2 exhibits a higher power grid strength, as shown in Fig.~\ref{fig:Dwshell15Hz} that the DW shell envelope of Case~2 has a larger $\alpha_{\rm LB}(j\omega)$ and is thus separated from the converter's DW shell. This is aligned with Fig.~\ref{fig:ieee-68bus-critical-frequency-range}~(b) where the stability distances become larger than zero.  

    Fig.~\ref{fig:ieee-68bus-time-domain} shows the
   responses of the system under the two settings, where a small disturbance occurs at $t=0.2~\mathrm{s}$. It can be seen that the system is unstable in Case~1 and it is stable in Case~2, consistent with the previous analysis based on DW shells and stability distances. Moreover, the oscillation frequency in Case~1 is within the critical frequency range where the stability distance is zero. Our approach can also identify the critical GFL converters that cause the instability. For instance, it is the GFL~5 and GFL~6 that result in the instability in Case~1, since the instability can be avoided by modifying their design so that their stability distances become larger than 0, just like the other GFL converters. 

    We remark that the curves in Fig.~\ref{fig:gfm-passivity-indices} are obtained using typical parameters of the devices, and one may achieve better GFM performance than those in Fig.~\ref{fig:gfm-passivity-indices}, e.g., by designing better control structures and employing better parameters. It can also be seen from Fig.~\ref{fig:gfm-passivity-indices} that although the SG has a high passivity index in the range of $[5~{\rm Hz},~12~{\rm Hz}]$, its passivity index within $[2~{\rm Hz},~5~{\rm Hz}]$ is worse than the GFM converters. As a final remark, the passivity index defined in this paper should be understood as a ``rotated'' passivity index since we multiply ${R}^{-1}(s)$ with the converter's admittance matrix in~\eqref{eq:rescaled-converter-admittance}, similar to the weighting matrix introduced in~\cite{chen2024extended}.


    
\begin{figure}[!t]
    \centering
    \includegraphics[width=3.1in]{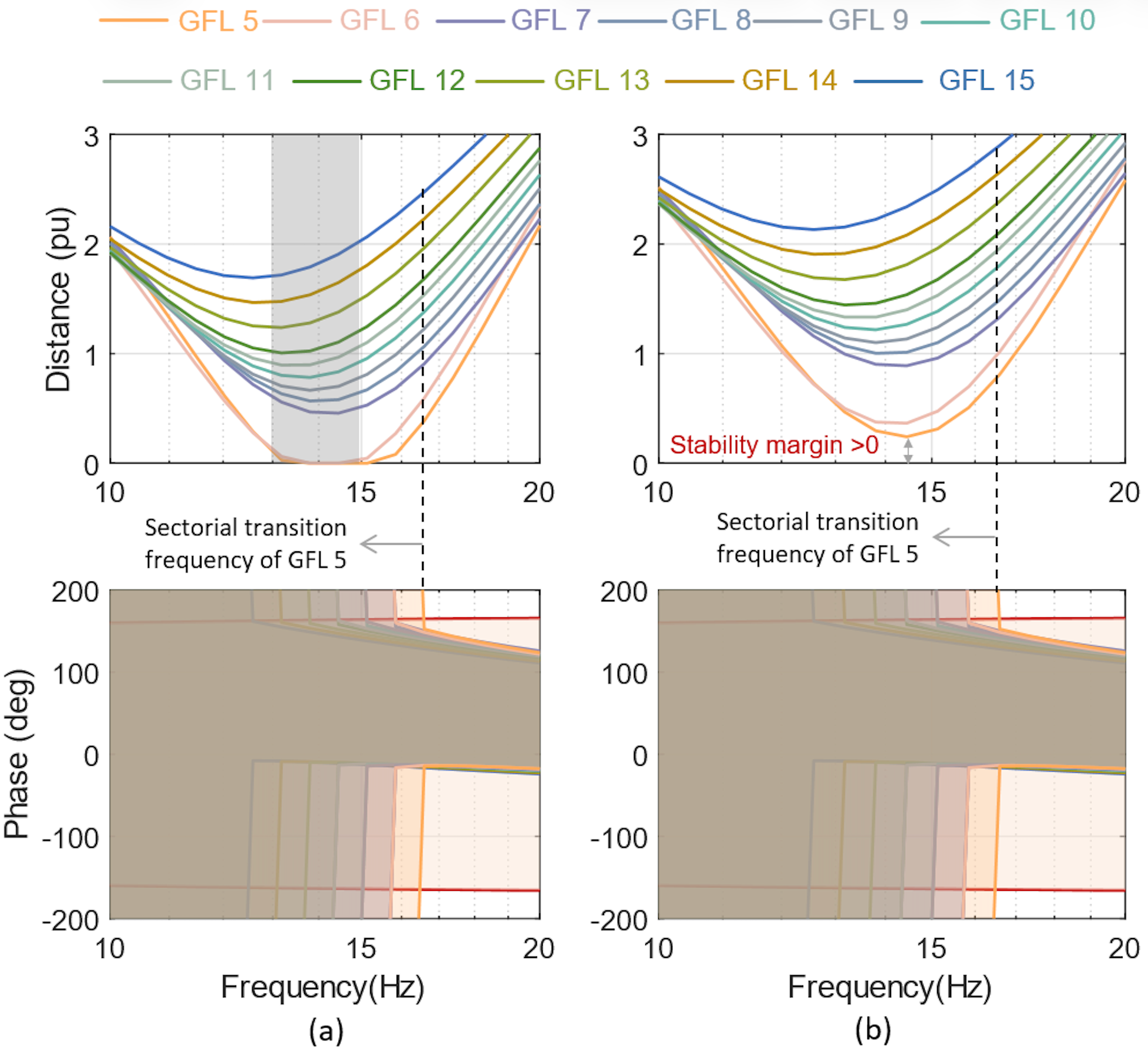}
    \vspace{-4mm}
    \caption{Zoomed-in frequency-wise stability distances and phase areas of the GFL converters and the equivalent network in the critical frequency range: (a) Case~1, and (b) Case~2.}
    \vspace{-2mm}
    \label{fig:ieee-68bus-critical-frequency-range}
\end{figure}

\begin{figure}[!t]
    \centering
    \includegraphics[width=2.8in]{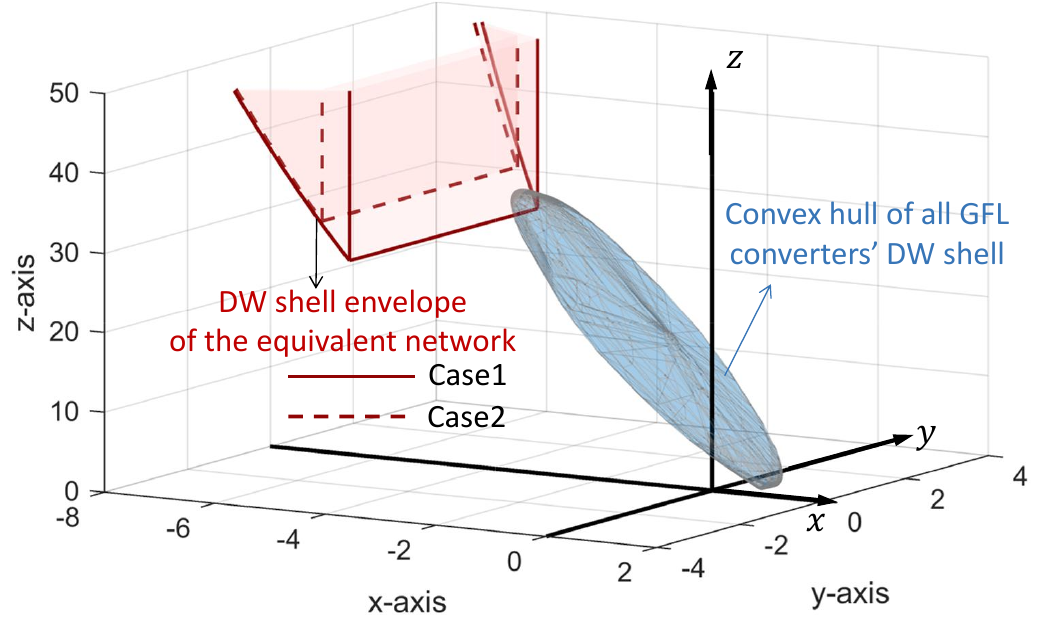}
    \vspace{-4mm}
    \caption{The DW shell envelope of the equivalent network and the convex hull of the GFL converters' DW shells at 14~Hz.}
    \label{fig:Dwshell15Hz}
     \vspace{-4mm}
\end{figure}


\begin{figure}[!t]
    \centering
    \includegraphics[width=2.7in]{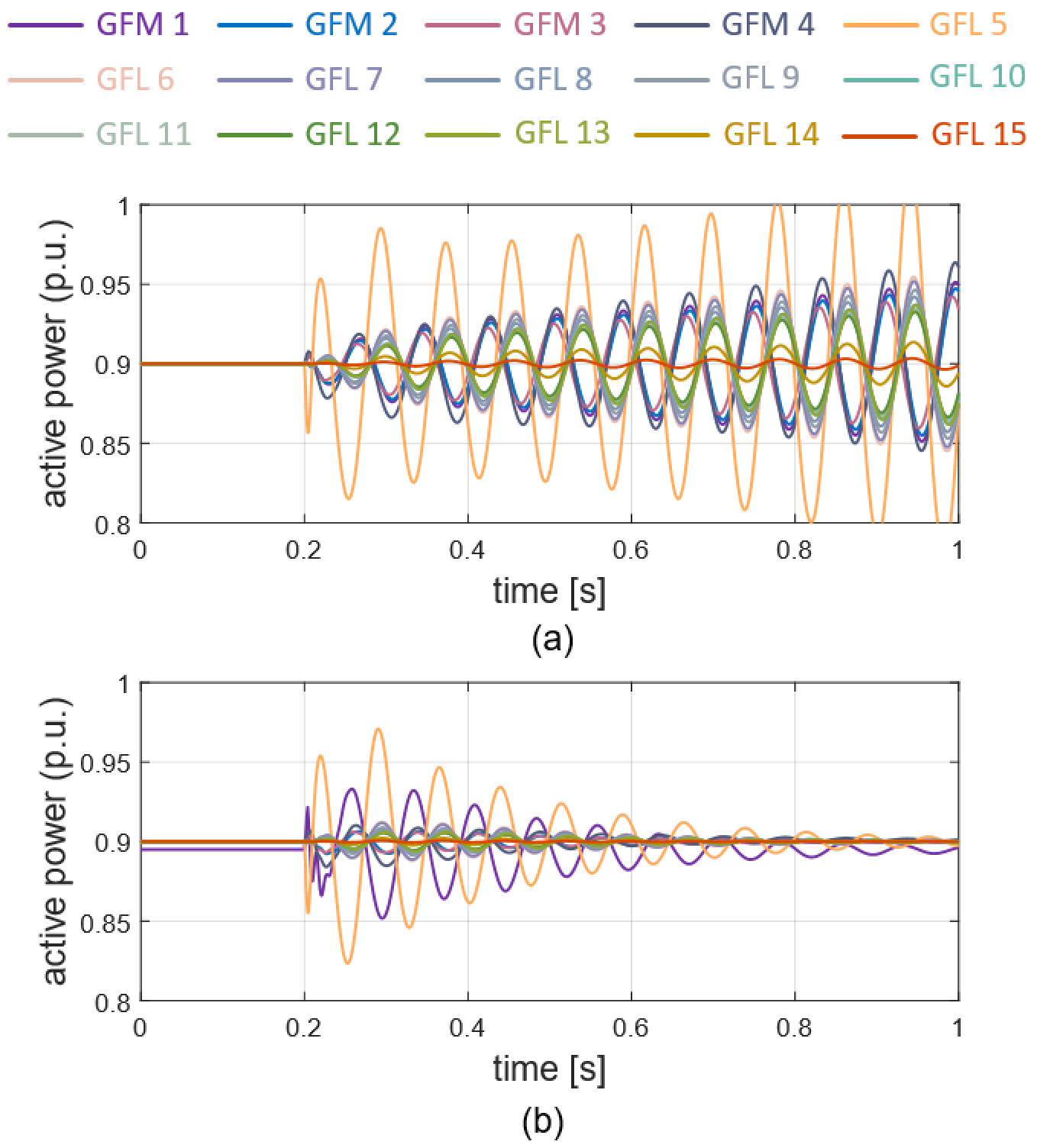}
    \vspace{-4mm}
    \caption{Time-domain responses of the modified IEEE 68-bus system: (a) Case~1, and (b) Case~2.}
    \vspace{-2mm}
    \label{fig:ieee-68bus-time-domain}
\end{figure}

     
    \section{Conclusions}\label{sec:conclusions}
    This paper developed a geometric framework for quantifying the stabilizing effects of heterogeneous GFM devices.
    Firstly, we derive an equivalent network reformulation which fuses the GFM dynamics into the network seen by the GFL converters. Then, we propose two GFM indices, including the passivity index and the imaginary-axis index, to capture the essential impact of GFM converters, and explicitly show how these two indices affect the DW shell envelope of the equivalent network and thus the overall system stability. For instance, we show that a higher passivity index increases the distance between the DW shell envelope of the equivalent network and the DW shells of the GFL converters, thereby increasing the power grid strength. 
    Our approach is scalable as it enables decentralized analysis of the interaction between the GFL converters and the equivalent network that incorporates the GFM dynamics. Moreover, it can handle heterogeneous GFM converters by focusing on the two GFM indices, which is suitable for analyzing large-scale converter-dominated power systems. 
    Future work will focus on the GFM controller synthesis problem based on DW shells.
    \vspace{-2mm}
    

\bibliographystyle{IEEEtran}
\bibliography{REFERANCE}
\numberwithin{equation}{section}
\vspace{0mm}

\newpage

\appendices
\numberwithin{equation}{section}
\vspace{0mm}
\section{Parameters of the Test Systems}

The main parameters of the three-converter system in Example 1 are as follows. The base value for power is 100 MVA and for frequency it is 50 Hz. Impedances $(Z_{i,j} = R_{i,j} + jX_{i,j})$ of the transmission lines(pu):
$Z_{1,4}= 0.002 + j0.0285,Z_{4,5}=0.001+j0.228,Z_{4,9}=0.005+j0.0855$,$Z_{5,6}=0.005+j0.114,Z_{3,6}=0.005+j0.1995,Z_{2,8}=0.003+j0.057,Z_{6,7}=0.002+j0.228,Z_{7,8}=0.003+j0.399,Z_{8,9}=0.002+j0.0285,Z_{9}=0.002+j0.0285$.

The main power network parameters of the modified IEEE 68-bus system are as follows. The load profile, shunt capacitors in the \(\Pi\)-model transmission lines, and generator capacities are the same as those in~[19]. The base power is \(100~\mathrm{MVA}\), and the rated frequency is \(50~\mathrm{Hz}\). The transmission-line impedances are denoted by \(Z_{i,j}=R_{i,j}+jX_{i,j}\), and the following values are given in \(10^{-2}~\mathrm{pu}\): \(Z_{1,54}=j0.0905\), \(Z_{2,58}=j0.1250\), \(Z_{3,62}=j0.1\), \(Z_{4,19}=0.0035+j0.071\), \(Z_{5,20}=0.0045+j0.09\), \(Z_{6,22}=j0.0715\), \(Z_{7,23}=0.0025+j0.136\), \(Z_{8,25}=0.003+j0.116\), \(Z_{9,29}=0.004+j0.078\), \(Z_{10,31}=j0.13\), \(Z_{11,32}=j0.065\), \(Z_{12,36}=j0.0375\), \(Z_{13,17}=j0.2475\), \(Z_{14,41}=j0.0075\), \(Z_{15,42}=j0.0075\), \(Z_{16,18}=j0.015\), \(Z_{17,36}=0.0025+j0.0225\), \(Z_{17,43}=0.0025+j0.138\), \(Z_{18,42}=0.002+j0.03\), \(Z_{18,49}=0.038+j0.5709\), \(Z_{18,50}=0.006+j0.144\), \(Z_{19,20}=0.0035+j0.069\), \(Z_{19,68}=0.008+j0.0976\), \(Z_{21,22}=0.004+j0.07\), \(Z_{21,68}=0.004+j0.0675\), \(Z_{22,23}=0.003+j0.048\), \(Z_{23,24}=0.011+j0.175\), \(Z_{24,68}=0.0015+j0.0295\), \(Z_{25,26}=0.016+j0.1615\), \(Z_{25,54}=0.035+j0.043\), \(Z_{26,27}=0.007+j0.0735\), \(Z_{26,28}=0.0215+j0.237\), \(Z_{26,29}=0.0285+j0.3125\), \(Z_{27,37}=0.0065+j0.0865\), \(Z_{27,53}=0.16+j1.6\), \(Z_{28,29}=0.007+j0.0755\), \(Z_{30,31}=0.0065+j0.0935\), \(Z_{30,32}=0.012+j0.144\), \(Z_{30,53}=0.004+j0.037\), \(Z_{30,61}=0.0047+j0.0458\), \(Z_{31,38}=0.0055+j0.0735\), \(Z_{31,53}=0.008+j0.0815\), \(Z_{32,33}=0.004+j0.0495\), \(Z_{33,34}=0.0055+j0.0785\), \(Z_{33,38}=0.018+j0.222\), \(Z_{34,35}=0.0005+j0.037\), \(Z_{34,36}=0.0165+j0.0555\), \(Z_{35,45}=0.0035+j0.0875\), \(Z_{36,61}=0.0055+j0.049\), \(Z_{37,52}=0.0035+j0.041\), \(Z_{37,68}=0.0035+j0.0445\), \(Z_{38,46}=0.011+j0.142\), \(Z_{39,44}=j0.2055\), \(Z_{39,45}=j0.4195\), \(Z_{40,41}=0.03+j0.42\), \(Z_{40,48}=0.01+j0.11\), \(Z_{41,42}=0.02+j0.3\), \(Z_{43,44}=0.0005+j0.0055\), \(Z_{44,45}=0.0125+j0.365\), \(Z_{45,51}=0.002+j0.0525\), \(Z_{46,49}=0.009+j0.137\), \(Z_{47,48}=0.0063+j0.067\), \(Z_{47,53}=0.0065+j0.094\), \(Z_{50,51}=0.0045+j0.1105\), \(Z_{52,55}=0.0055+j0.0665\), \(Z_{53,54}=0.0175+j0.2055\), \(Z_{54,55}=0.0065+j0.0755\), \(Z_{55,56}=0.0065+j0.1065\), \(Z_{56,57}=0.004+j0.064\), \(Z_{56,66}=0.004+j0.0645\), \(Z_{57,58}=0.001+j0.013\), \(Z_{57,60}=0.004+j0.056\), \(Z_{58,59}=0.003+j0.046\), \(Z_{58,63}=0.0035+j0.041\), \(Z_{59,60}=0.002+j0.023\), \(Z_{60,61}=0.0115+j0.1815\), \(Z_{62,63}=0.002+j0.0215\), \(Z_{62,65}=0.002+j0.0215\), \(Z_{63,64}=0.008+j0.2175\), \(Z_{64,65}=0.008+j0.2175\), \(Z_{65,66}=0.0045+j0.0505\), \(Z_{66,67}=0.009+j0.1085\), and \(Z_{67,68}=0.0045+j0.047\).

\newpage

\section{Parameters of the Devices}

See Table~\ref{tab:main_parameters}.

\begin{table}[!t]
\centering
\caption{Main parameters of the converters and synchronous generator}
\label{tab:main_parameters}
\renewcommand{\arraystretch}{1.12}
\setlength{\tabcolsep}{4pt}
\small

\begin{tabular}{|p{0.48\columnwidth}|c|c|}
\hline
\multicolumn{3}{|c|}{\textbf{Parameters of the GFL Converter}} \\
\hline
Parameter & Symbol & Value \\
\hline
Filter inductance
& $L_F$ & $0.05$ p.u. \\
\hline
Filter capacitance
& $C_F$ & $0.06$ p.u. \\
\hline
Grid-side inductance
& $L_g$ & $0.05$ p.u. \\
\hline
Current-loop PI gains
& $\{K_{p,i},K_{i,i}\}$ & $\{0.3,10\}$ \\
\hline
Voltage-feedforward time constant
& $T_{\mathrm{ff}}$ & $0.02$ s \\
\hline
Power-loop PI gains
& $\{K_{p,pq},K_{i,pq}\}$ & $\{0.5,40\}$ \\
\hline
PLL bandwidth
& $\omega_{\mathrm{PLL}}$ & $40$ rad/s \\
\hline

\multicolumn{3}{|c|}{\textbf{Parameters of the VSG-Controlled GFM Converter}} \\
\hline
Parameter & Symbol & Value \\
\hline
Filter inductance
& $L_F$ & $0.05$ p.u. \\
\hline
Filter capacitance
& $C_F$ & $0.06$ p.u. \\
\hline
Grid-side inductance
& $L_g$ & $0.15$ p.u. \\
\hline
Current-loop PI gains
& $\{K_{p,i},K_{i,i}\}$ & $\{0.3,10\}$ \\
\hline
Voltage-feedforward time constant
& $T_{\mathrm{ff}}$ & $0.02$ s \\
\hline
Voltage-loop PI gains
& $\{K_{p,v},K_{i,v}\}$ & $\{2,10\}$ \\
\hline
Virtual inertia coefficient
& $J$ & $2$ \\
\hline
Damping coefficient
& $D$ & $50$ \\
\hline

\multicolumn{3}{|c|}{\textbf{Parameters of the Droop-Controlled GFM Converter}} \\
\hline
Parameter & Symbol & Value \\
\hline
Filter inductance
& $L_F$ & $0.05$ p.u. \\
\hline
Filter capacitance
& $C_F$ & $0.06$ p.u. \\
\hline
Grid-side inductance
& $L_g$ & $0.15$ p.u. \\
\hline
Current-loop PI gains
& $\{K_{p,i},K_{i,i}\}$ & $\{0.3,10\}$ \\
\hline
Voltage-feedforward time constant
& $T_{\mathrm{ff}}$ & $0.02$ s \\
\hline
Voltage-loop PI gains
& $\{K_{p,v},K_{i,v}\}$ & $\{2,10\}$ \\
\hline
Active-power droop coefficient
& $k_p$ & $50$ \\
\hline

\multicolumn{3}{|c|}{\textbf{Parameters of the DC-Link-Synchronized GFM Converter}} \\
\hline
Parameter & Symbol & Value \\
\hline
Filter inductance
& $L_F$ & $0.05$ p.u. \\
\hline
Filter capacitance
& $C_F$ & $0.06$ p.u. \\
\hline
Grid-side inductance
& $L_g$ & $0.15$ p.u. \\
\hline
Current-loop PI gains
& $\{K_{p,i},K_{i,i}\}$ & $\{0.3,10\}$ \\
\hline
Voltage-feedforward time constant
& $T_{\mathrm{ff}}$ & $0.02$ s \\
\hline
Voltage-loop PI gains
& $\{K_{p,v},K_{i,v}\}$ & $\{2,10\}$ \\
\hline
Synchronization coefficient
& $K_T$ & 4 \\
\hline
Inertia coefficient
& $K_J$ & 10 \\
\hline
Damping coefficient
& $K_D$ & 750 \\
\hline

\multicolumn{3}{|c|}{\textbf{Parameters of the Synchronous Generator}} \\
\hline
\multicolumn{2}{|l|}{Synchronous-generator parameters}
& Ref.~[19] \\
\hline
\end{tabular}
\end{table}

\end{document}